\documentclass[a4paper,twocolumn,11pt]{quantumarticle}
\pdfoutput=1
\usepackage[utf8]{inputenc}
\usepackage[english]{babel}
\usepackage[T1]{fontenc}

\usepackage{amsmath,amssymb}
\usepackage{amsthm}
\usepackage{hyperref} 
\usepackage{acronym}  
\usepackage{color}  
\usepackage{amsfonts}
\usepackage{mathrsfs}
\usepackage{graphicx}
\usepackage[on]{psfrag}
\usepackage{array}
\usepackage{cases}
\usepackage{tabularx}
\usepackage{url}
\usepackage{soul}
\usepackage{theoremref}
\usepackage{thmtools}
\usepackage{lettrine}
\usepackage{mathtools}
\usepackage{dsfont}
\usepackage[table]{xcolor}
\usepackage{multirow}
\usepackage{dashbox}
\usepackage{algorithm}
\usepackage{algpseudocode}
\usepackage{float}
\usepackage{enumitem}
\usepackage{hhline}
\usepackage{empheq}
\usepackage{dashrule}
\usepackage{threeparttable}
\usepackage[numbers]{natbib}

\usepackage{tikz}
\usepackage{lipsum}

\usepackage{xpatch}

\makeatletter
\renewcommand{\ALG@name}{Protocol}
\makeatother

\usepackage{relsize}%
\usepackage{Winsnotation}

\newcommand{\SU}{\hspace{0.3mm}\Set U}
\newcommand{\SM}{\hspace{-0.15mm}\Set M}

\makeatletter
\renewcommand*\env@matrix[1][\arraystretch]{%
  \edef\arraystretch{#1}%
 \hskip -\arraycolsep
  \let\@ifnextchar\new@ifnextchar
  \array{*\c@MaxMatrixCols c}}
\makeatother

\makeatletter
    \def\CT@@do@color{%
      \global\let\CT@do@color\relax
            \@tempdima\wd\z@
            \advance\@tempdima\@tempdimb
            \advance\@tempdima\@tempdimc
    \advance\@tempdimb\tabcolsep
    \advance\@tempdimc\tabcolsep
    \advance\@tempdima2\tabcolsep
            \kern-\@tempdimb
            \leaders\vrule
                    \hskip\@tempdima\@plus  1fill
            \kern-\@tempdimc
            \hskip-\wd\z@ \@plus -1fill }
\makeatother

\makeatletter
\def\hlinew#1{%
  \noalign{\ifnum0=`}\fi\hrule \@height #1 \futurelet
   \reserved@a\@xhline}
\newcolumntype{"}{@{\hskip\tabcolsep\vrule width 0.8pt\hskip\tabcolsep}}   
\makeatother

\newcommand*{\argmin}{\operatornamewithlimits{argmin}\limits}

\newlength{\Oldarrayrulewidth}

\newcolumntype{L}[1]{>{\raggedright\let\newline\\\arraybackslash\hspace{0pt}}m{#1}}
\newcolumntype{C}[1]{>{\centering\let\newline\\\arraybackslash\hspace{0pt}}m{#1}}
\newcolumntype{R}[1]{>{\raggedleft\let\newline\\\arraybackslash\hspace{0pt}}m{#1}}

\definecolor{dblue}{rgb}{0,0,0.75}
\definecolor{dgreen}{rgb}{0,0.7,0}
\definecolor{lightblue}{rgb}{0.83,0.91,1.0}
\newcommand{\Clb}{\cellcolor{lightblue}}

\declaretheoremstyle[
spaceabove=4pt, spacebelow=4pt,
headfont=\normalfont\bfseries,
notefont=\normalfont, notebraces={(}{)},
headpunct={:},
bodyfont=\normalfont,
postheadspace=0.5em,
]{mynote}
\declaretheoremstyle[
spaceabove=4pt, spacebelow=4pt,
headfont=\normalfont,
notefont=\normalfont, notebraces={(}{)},
headpunct={:},
bodyfont=\normalfont,
postheadspace=0.5em,
]{mynote2}

\declaretheorem[style=mynote]{Remark}

\declaretheorem[style=mynote,name=Theorem]{Thm}

\declaretheorem[style=mynote,name=Theorem]{ThmA}
\declaretheorem[style=mynote,name=Lemma]{LemA}
\declaretheorem[style=mynote,name=Corollary]{CorA}

\acrodef{QED}{quantum entanglement distillation}
\acrodef{w.r.t.}{with respect to}
\acrodef{i.i.d.}{independent and identically distributed}
\acrodef{pmf}{probability mass function}
\acrodef{pdf}{probability density function}
\acrodef{cdf}{cumulative distribution function}
\acrodef{POVM}{positive operator-valued measure}
\acrodef{QBER}{quantum bit error rate}
\acrodef{QKD}{quantum key distribution}

\newcommand*{\QEDB}{\hfill\ensuremath{\square}}%

\begin{document}

\title{\parbox{17cm}{ Reliable Interval Estimation for the Fidelity of Entangled
 States in Scenarios with General Noise} }

\author{Liangzhong Ruan}
\affiliation{
School of Cyber Science and Engineering, Xi'an Jiaotong University, Xi'an, China
}
\orcid{0000-0002-5147-1289}
\email{liangzhongruan@xjtu.edu.cn}

\author{Bas Dirkse}
\affiliation{Qutech, Delft University of Technology, Delft, Netherlands}
\orcid{0000-0002-4391-6962}
\email{B.Dirkse@tudelft.nl}
\maketitle

\begin{abstract}
Fidelity estimation for entangled states constitutes an essential building block for quality control and error detection in quantum networks. Nonetheless, quantum networks often encounter heterogeneous and correlated noise, leading to excessive uncertainty in the estimated fidelity. In this paper, the uncertainty associated with the estimated fidelity under conditions of general noise is constrained by jointly employing random sampling, a thought experiment, and Bayesian inference, resulting in a credible interval for fidelity that is valid in the presence of general noise. The proposed credible interval incorporates all even moments of the posterior distribution to enhance estimation accuracy. Factors influencing the estimation accuracy are identified and analyzed. Specifically, the issue of excessive measurements is addressed, emphasizing the necessity of properly determining the measurement ratio for fidelity estimation under general noise conditions.
\end{abstract}

\section{Introduction}

Entanglement distributed in quantum networks \cite{Suzetal:J17} is a valuable resource that facilitates a variety of promising quantum applications such as secure quantum communication \cite{Zhouetal:J19}, distributed quantum computing \cite{SerManBos:J06}, enhanced sensing \cite{Xiaetal:J23}, and clock synchronization \cite{Kometal:J14}.
Notably, entanglement serves as the primary source of quantum advantage in the emerging quantum internet \cite{WehElkHan:J18,Panetal:J19,Lietal:J22}.
In this context, quantifying the quality of entangled states shared by network nodes 
is imperative for quantum networks.
To this end, fidelity estimation for entangled states is a promising candidate.
Fidelity quantifies the quality of quantum states \cite{Basetal:J11,ArrLazHelBal:J14,ZhaLuoWen-etal:J21} 
and can be estimated with separable quantum measurements and classical post-processing.
Fidelity estimation protocols for several types of states are designed in \cite{SomChiBer:J06,GuhLuGaoPan:J07,ZhuHay:J19}, 
a fidelity estimation protocol for general pure states using only Pauli observables is proposed in \cite{FlaLiu:J11},
and a low-complexity fidelity estimation protocol for mixed states is proposed in \cite{WanZhaCheetal:J23}.
Fidelity estimation can also be achieved by using quantum fingerprinting \cite{BuhCleWatWol:J01},
quantum state certification \cite{BadDonWri:J19},  quantum hypothesis testing \cite{HayMatTsu:J06,Hay:J09} and distributed inner product estimation \cite{AnsLanLiu:J22}.
Fidelity estimation protocols have been successfully  implemented in several recent experiments \cite{Luetal:J15,Kaletal:J17,Humetal:J18,Pom-etal:J21}.

Quantum networks often encounter heterogeneous and correlated noise. 
Nonetheless, the impact of such noise on fidelity estimation has yet to be thoroughly addressed.
As has been noted in the context of quantum key distribution rate characterization \cite{TomLimGisRen:J12}
and quantum channel capacity estimation \cite{PfiRolManTomWeh:J18}, noise that is not \ac{i.i.d.} can lead to atypical measurement outcomes, 
thereby introducing excessive uncertainty into the estimation process.
Similarly, managing the uncertainty in the estimated fidelity of entangled states under conditions of general noise presents a significant challenge in achieving reliable quality control for quantum networks.

This paper addresses the aforementioned challenge. 
Given that measurements collapse quantum states, we consider a setup in which the two nodes sample a subset of qubit pairs for measurement 
and estimate the average fidelity of unsampled pairs conditioned on the measurement outcome.
With this setup, \cite{Ruan:J23} proposes a method for efficiently generating a point estimate in scenarios with general noise.
Nevertheless, a point estimate alone is insufficient to characterize the uncertainty associated with the estimated fidelity. 
The main contribution of this paper are as follows:
\begin{itemize}
\item Employing random sampling, a thought experiment, and Bayesian inference to establish bounds on the effects of non-\ac{i.i.d.} noise and obtain a reliable interval estimate applicable to scenarios with general noise;
\item Characterizing factors that affect the estimation accuracy in scenarios with general noise, namely the number of measurements, the tightness of tail probability bound, and the error due to atypical quantum states;
\item Identifying the issue of excessive measurements, a phenomenon in which more measurements lead to less accurate estimates,
and providing a simple condition to prevent this undesirable phenomenon.
\end{itemize}

The remainder of the paper is organized as follows. Section~\ref{sec:protocol} details the system setup and the proposed protocol. 
Section~\ref{sec:theory} elaborates on the theoretical foundation underpinning the proposed protocol. 
Section~\ref{sec:factors} analyzes the factors that influence estimation accuracy under conditions of general noise. Section~\ref{sec:sim} outlines the numerical studies. 
Finally, Section~\ref{sec:conclusion} provides the conclusion.

Random variables and their realizations are denoted by upper- and lower-case letters respectively, e.g., $F$ and $f$.

\section{Proposed Protocol}
\label{sec:protocol}
This section describes the system setup, the proposed protocol, 
and a demonstrative application.
\subsection{System Setup}
Two remote nodes aim to share entangled qubit pairs over a noisy quantum channel.
All maximally entangled states are mutually convertible via local unitary operators.
Hence, without loss of generality,
we assume that the nodes aim to share qubit pairs in state $|\Psi^-\rangle$, where
\begin{align}
|\Psi^\pm \rangle =\frac{|01\rangle \pm |10\rangle}{\sqrt{2}}, \quad |\Phi^\pm \rangle =\frac{|00\rangle \pm |11\rangle}{\sqrt{2}}
\end{align}
are the Bell states.

The nodes share $N$ entangled qubit pairs with joint state $\V{\rho}_{\mathrm{all}}$.
In this paper, we consider general noise with no predefined statistical properties, i.e., $\V{\rho}_{\mathrm{all}}$ can be any $2N$-qubit state and the nodes have no prior information about $\V{\rho}_{\mathrm{all}}$.
To evaluate the quality of entanglement, the nodes sample $M(<N)$ number of qubit pairs to measure.
Specifically, denote the set of all qubit pairs by $\Set N$, with $|\Set N| = N$. 
In this case, the set of sampled pairs, $\Set{M}$, is drawn from all $M$-subsets of $\Set N$ with equal probability.

The nodes perform measurements according to Protocol~\ref{alg:fidelityest}, which was proposed in \cite{Ruan:J23} and minimizes the variance of the estimated fidelity.
Because measurements collapse quantum states, the nodes estimate the average fidelity of the unsampled pairs given the measurement outcome $\V{r}=\{r_n, n\in \SM\}$.
Specifically, denote $\V{\rho}_{\mathrm{all}}^{(\V{r})}$ as the joint state of the qubit pairs conditioned on the measurement outcome $\V{r}$.
In this case, the state of the $n$-th pair is
\begin{align}
\V{\rho}^{(\V{r})}_n = \mathrm{Tr}_{i\in{\Set N}\backslash \{n\}} \V{\rho}_{\mathrm{all}}^{(\V{r})},
\end{align}
and the average fidelity of the unsampled qubit pairs \ac{w.r.t.} the target state $|\Psi^-\rangle$ is 
\begin{align}
\bar{f}= \frac{1}{N-M}\sum_{n\in \Set N \backslash \SM}\langle\Psi^-|\V{\rho}^{(\V{r})}_n |\Psi^-\rangle.\label{eqn:Fbar}
\end{align}

To achieve the design goals, many quantum applications require that the fidelity of entangled states is above certain thresholds with a high probability \cite{WehElkHan:J18,Zhangetal:J22}.
To this end, one must obtain interval estimates for the fidelity.
In this paper, the interval estimate is given by the credible interval $\Set C$, the interval in which the actual fidelity lies with probability $\alpha$ conditioned on the measurement outcome $\V{r}$, i.e.,
\begin{align}
\Pr\big[\bar{F} \in {\Set C} \big| \V{r} \big]= \alpha,
\end{align}
where $\bar{F}$ denotes the average fidelity $\bar{f}$ considered as a random variable.
For a discussion on the choice of statistical paradigm, see  Appendix~\ref{sec:method}.

\subsection{Protocol Description}

\setcounter{algorithm}{-1}
\begin{figure}
\vspace{-2mm}
\begin{algorithm}[H]
\caption{Measurement operation}\label{alg:fidelityest}
\begin{algorithmic}[1]
\State{\em Preset measurement parameters.} The nodes select the sample set $\SM$ completely at random and generate $M$ number of \ac{i.i.d.} random variables $A_n \in\{x,y,z\}$, $n\in \SM$,
with distribution 
$
\Pr[A_n=u]=\frac{1}{3}, \quad u\in\{x,y,z\}.
$
\State{\em Perform measurements.}
For the qubit pair $n\in \Set{M}$,  both nodes measure the qubit in the $A_n$-basis.
If the measurement results of the two nodes match, record measurement outcome  $r_{n} = 1$, otherwise record $r_{n} = 0$.
\end{algorithmic}
\end{algorithm}
\end{figure}

\begin{figure}
\vspace{-2mm}
\begin{algorithm}[H]
\caption{Interval estimate of fidelity}\label{alg:fidelityest_int}
\begin{algorithmic}[1]
\State{\em Process the measurement data.} 
Take the measurement outcomes $\{r_n, n\in \SM\}$ of Protocol~\ref{alg:fidelityest} as input, and
calculate the number of errors and and the \ac{QBER}
\begin{align}
e_{\Set M} = \sum_{n\in\Set M} r_{n}, \quad \varepsilon_{\Set M}= \dfrac{e_{\Set M}}{M}.
\end{align} 

\State{\em Characterize posterior distributions.} 
Calculate the optimal computed moment $T^*$: 
\begin{align}
T^* &= 2\bigg\lfloor\frac{-2\ln(1-\alpha) + 0.8}{2}\bigg\rceil,\label{eqn:Tstar}
\end{align}
and the $2t$-th central moments $M^{(2t)}_{\mathrm{c}}$, $t\in\{1,2,\ldots,\frac{T^*}{2}\}$, according to 
\begin{align}
M^{(2t)}_{\mathrm{c}}
&=\sum_{e=0}^{N-M}{\Set P}(e;{e_{\SM}})
\bigg(\frac{e}{N-M} - \varepsilon_{\SM} - \dfrac{\frac{1}{2}-\varepsilon_{\SM}}{M+1}\bigg)^t, \label{eqn:emoments}
\end{align}
\Statex where ${\Set P}(e;e_{\SM})$, the distribution function of $e$ parameterized by $e_{\SM}$, is expressed as
\begin{align}
{\Set P}(e;e_{\SM}) = & \bigg(\!\begin{array}{c}N\!-\! M\\ e\end{array}\!\bigg)\nonumber\\
&\frac{\mathcal{B}(e+e_{\SM}+\frac{1}{2},N-e- e_{\SM}+\frac{1}{2})}{\mathcal{B}(e_{\SM}+\frac{1}{2},M- e_{\SM}+\frac{1}{2})},\label{eqn:dis_embar}
\end{align} 
in which $\mathcal{B}(a,b)=\int_{0}^1 t^{a-1} (1-t)^{b-1} dt$ is the beta function.
\State {\em Output the interval estimate.}
Obtain the credible interval
\begin{align}
{\Set C}^{(T^*)} =&\big[\tilde{f}-\delta^{(T^*)}, \tilde{f}+\delta^{(T^*)} \big],\label{eqn:CT_pro}
\end{align}
\Statex where the center point of the interval
\begin{align}
\tilde{f}= 1- \dfrac{3}{2}\bigg(\varepsilon_{\Set M} + \dfrac{\frac{1}{2}-\varepsilon_{\Set M}}{M+1}\bigg),\label{eqn:tildeF}
\end{align}
and the radius of the interval
\begin{align}
\delta^{(T^*)} = \min_{\scriptsize t\in\{1,2,\ldots,\frac{T^*}{2}\} }\bigg\{\frac{3}{2}\bigg(\frac{M^{(2t)}_{\mathrm{c}}}{1-\alpha}\bigg)^{\frac{1}{2t}} \bigg\}.\label{eqn:deltastar-m}
\end{align}
\end{algorithmic}
\end{algorithm}
\end{figure}

Protocol~\ref{alg:fidelityest_int} processes the measurement outcome $\V{r}$ and returns the credible interval ${\Set C}^{(T^*)}$, which is reliable in scenarios with general noise.
The theoretical basis of this protocol is given in Section~\ref{sec:theory}, with the formal theorems and detailed derivations outlined in  Appendix~\ref{sec:Bayesian}. 
In particular, the main result is summerized by \thref{thm:cred_interval_m}.
\begin{Thm}[Interval estimation]\thlabel{thm:cred_interval_m}
In cases with general noise, the average fidelity $\bar{F}$ lies in the interval ${\Set C}^{(T^*)}$ defined in \eqref{eqn:CT_pro} with probablity at least $\alpha$ conditioned on the \ac{QBER} $\varepsilon_{\SM}$.
\end{Thm}

\begin{Remark}[Computational complexity of Protocol~\ref{alg:fidelityest_int}]
Computing the radius of the credible interval $\delta^*$ represents the main computational load of Protocol~\ref{alg:fidelityest_int}.
According to \eqref{eqn:deltastar-m}, the complexity of computing $\delta^*$ is determined by the computational load for calculating $M^{(2t)}_{\mathrm{c}}$, where $t\in\{1,2,\ldots,\frac{T^*}{2}\}$.

According to  \eqref{eqn:Tstar}, $T^*$ scales as $\Set{O}(-\ln(1-\alpha))$.
According to \eqref{eqn:emoments}, the complexity of computing $M^{(2t)}_{\mathrm{c}}$ scales linearly \ac{w.r.t.} the number of  unsampled qubit pairs, i.e., $\Set{O}(M-N)$.
Therefore, the overall complexity of Protocol~\ref{alg:fidelityest_int} scales at 
\begin{align}
\Set{O}(-\ln(1-\alpha)(M-N)),
\end{align}
which is a mild computational load.~\hfill~\qed
\end{Remark}

\subsection{Demonstrative application}
To illustrate the performance of the proposed protocol in scenarios with non-\ac{i.i.d.} noise, we consider a demonstrative example to compare the effects of two types of noise, i.e., \ac{i.i.d.} noise, with joint state of all qubit pairs
\begin{align}
\V{\rho}_{\mathrm{all}}=\otimes^{N} \big(\bar{f}_{\mathrm{all}}|\Psi^-\rangle\langle\Psi^-| 
+ (1-\bar{f}_{\mathrm{all}})|00\rangle\langle00|\big),
\end{align}
and heterogeneous noise, with joint state
\begin{align}
\V{\rho}_{\mathrm{all}}=\big(\otimes^{N\bar{f}_{\mathrm{all}}} |\Psi^-\rangle\langle\Psi^-|\big)\otimes \big(\otimes^{N(1-\bar{f}_{\mathrm{all}})}|00\rangle\langle00|\big),
\end{align}
where the average fidelity of all qubit pairs $\bar{f}_{\mathrm{all}}$ is selected such that $N\bar{f}_{\mathrm{all}}$ is an integer.

\begin{figure}[t] \centering
\includegraphics[scale=0.95]{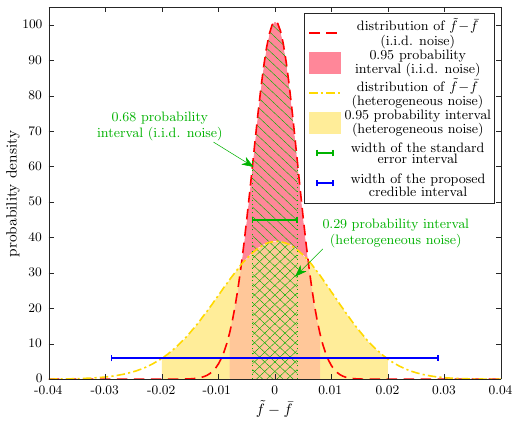}
\caption {A comparison of the standard error interval, i.e., the interval of mean plus-minus standard deviation, the credible interval obtained by the proposed protocol, and the distribution of the estimation error $\tilde{f}-\bar{f}$.
In this figure, $\bar{f}_{\mathrm{all}}=0.9$, $\alpha=0.95$. 
}
\label{fig_TvsInt_Exa}
\end{figure}

As Fig.~\ref{fig_TvsInt_Exa} shows, in cases with \ac{i.i.d.} or heterogeneous noise, the true fidelity is within the standard error interval  with a probability of $0.68$ and $0.29$, respectively.
This contrast shows that in cases with non-\ac{i.i.d.} noise, the standard error interval is no longer a reliable interval estimator.
The credible interval obtained by Protocol~\ref{alg:fidelityest_int} contains the $\alpha$ probability intervals of the actual fidelity for both types of noise, demonstrating its reliability.

\section{Interval Estimation with General Noise}
\label{sec:theory}
This section proves that the credible interval generated by Protocol~\ref{alg:fidelityest_int} is legitimate in the presence of general noise.
The main text explains the key steps and concepts.
The formal theorems and detailed derivations are in  \ref{sec:Bayesian}.

In the presence of non-\ac{i.i.d.} noise, the estimation error is unlikely to be Gaussian.
In such scenarios, conventional methods based on the second-order moment lead to inaccurate interval estimates \cite{LinBas:J00,XiHuYou:J12}.
To overcome this challenge, we construct a thought experiment to utilize all even moments of the posterior distribution and obtain an accurate interval estimate.

In the thought experiment, all unsampled qubit pairs are measured after the sampled pairs.
Denote the set of unsampled qubit pairs by
\begin{align}
\SU = \Set N \backslash \SM,
\end{align}
and denote the \ac{QBER} of unsampled qubit pairs by
\begin{align}
\mathcal{E}_{\SU} = \frac{\sum_{n\in {\SU}} R_{n}}{N-M}.
\end{align}

The thought experiment enables a Bayesian inference on the average fidelity $\bar{F}$ based on the following two facts.
\begin{itemize}
\item[A.] Because the sampling set $\Set M$ is chosen completely at random, conditioned on a realization of $\mathcal{E}_{\Set M}=\varepsilon_{\Set M}$, the posterior distribution of $\mathcal{E}_{\SU}$ is concentrated around $\varepsilon_{\Set M}$;
\item[B.] Because the expected value of $\mathcal{E}_{\SU}$ is a linear function of the average fidelity ${\bar{F}}$, the central moments of the posterior distribution of $\mathcal{E}_{\SU}$ can be exploited to bound those of  ${\bar{F}}$.
\end{itemize}

The credible interval of ${\bar{F}}$ is derived based on the above facts. 
The derivation involves two main steps, namely, obtaining moments of the posterior distribution of $\mathcal{E}_{\SU }$ and 
bounding the tail probability of~$\bar{F}$. 

\subsection{Characterize the moments of the posterior distribution} 
The first step characterizes the moments of the posterior distribution of $\mathcal{E}_{\SU}$, i.e., $\Set{P}(\varepsilon_{\SU}|\varepsilon_{\SM})$.

According to Bayes' theorem, to derive the posterior distribution $\Set{P}(\varepsilon_{\SU }|\varepsilon_{\SM})$,
it is necessary to determine the likelihood function $\Set{P}(\varepsilon_{\SM}|\varepsilon_{\SU})$.
However, the likelihood function is difficult to obtain in the presence of general noise.

To overcome the above challenge, we reconstruct the thought experiment by exploiting the commutativity of measurements on different qubit pairs.
Specifically, due to their commutativity, the order of measurements does not change the outcome $r_n\in\{0,1\}$, $n\in \Set N$.
Therefore, one can perform the thought experiment in a sequence that makes all measurements before sampling.
Subsequently, the sampling is done from a set of classical binary variables $r_n$, $n\in \Set N$, regardless of the statistical properties of the noise.
For such a classical random sampling process, the posterior distribution $\Set{P}(\varepsilon_{\SU }|\varepsilon_{\SM})$ and its moments can be obtained using Bayes' theorem, the properties of beta-binomial distributions, and the relationship between the raw and central moments.

The formal propositions of this step are presented as \thref{lem:unsampled} and \thref{lem:moments} in  Appendix~\ref{pf_lem:unsampled}.

\subsection{Bound the tail probability of fidelity}
The second step bounds the tail probability of average fidelity $\bar{F}$ based on the moments of the posterior distribution $\Set{P}(\varepsilon_{\SU }|\varepsilon_{\SM})$.

In the case with general noise, the likelihood function $\Set{P}(\varepsilon_{\SM} |\bar{f})$ of the average fidelity $\bar{F}$ is unknown.
Consequently, the interval estimate of $\bar{F}$ cannot be obtained by characterizing its posterior distribution.
A pivotal concept to overcome this challenge is that because the average fidelity ${\bar{F}}$ is a linear function of the expected value of $\mathcal{E}_{\SU}$,
the randomness of ${\bar{F}}$ shall be no more than that of the linear function of $\mathcal{E}_{\SU }$.
Consequently, the moments of the posterior distribution $\Set{P}(\varepsilon_{\SU }|\varepsilon_{\SM})$ can be used to bound the tail probability of ${\bar{F}}$.

More precisely, by applying the analysis in \cite[Lemma B.4]{Ruan:J23} to the measurement outcomes of the thought experiment, i.e., $\{r_n, n\in \SU = \Set N \backslash \SM\}$, the following relationship between the average fidelity ${\bar{F}}$ and \ac{QBER} $\mathcal{E}_{\SU }$ of the unsampled qubit pairs can be obtained.
\begin{align}
\bar{f} = 1-\frac{3}{2}\mathbb{E}\big[\mathcal{E}_{\SU }\big|\bar{f},\varepsilon_{\SM} \big].\label{eqn:ktoM-m}
\end{align}
According to \eqref{eqn:ktoM-m} and the tower property of the conditional expectation \cite[P.4]{Pit:15},
\begin{align}
\mathbb{E}[\bar{F}|\varepsilon_{\SM}] 
&= \mathbb{E}\bigg[1-\frac{3}{2}\mathbb{E}\big[\mathcal{E}_{\SU }\big|\bar{f},\varepsilon_{\SM} \big]\bigg|\varepsilon_{\SM}\bigg]\nonumber\\
&=1-\frac{3}{2}\mathbb{E}\big[\mathcal{E}_{\SU }\big|\varepsilon_{\SM} \big].\label{eqn:equalexp-m}
\end{align}
According to \eqref{eqn:ktoM-m} and \eqref{eqn:equalexp-m},  the $2t$-th central moment of the posterior distribution of $\bar{F}$, $\mathbb{M}^{(2t)}_{\mathrm{c}}\big[\bar{F}|\varepsilon_{\SM} \big]$, is upper bounded as follows.
\begin{subequations}
\begin{align}
&\mathbb{M}^{(2t)}_{\mathrm{c}}\big[\bar{F}|\varepsilon_{\SM} \big] \nonumber\\
&= \mathbb{E}\Big[\big(\bar{F}-\mathbb{E}[\bar{F}|\varepsilon_{\SM}]\big)^{2t}\Big|\varepsilon_{\SM}\Big] \label{Mc2t-F-E-a}\\
&=\bigg(\frac{3}{2}\bigg)^{2t} \mathbb{E}\Big[\big(\mathbb{E}[\Set{E}_{\SU }|\bar{f},\varepsilon_{\SM}]
-\mathbb{E}[\Set{E}_{\SU }|\varepsilon_{\SM}]\big)^{2t}\Big|\varepsilon_{\SM}\Big]
\label{Mc2t-F-E-b}\\
&=\bigg(\frac{3}{2}\bigg)^{2t} \mathbb{E}\Big[\big(\mathbb{E}\big[\Set{E}_{\SU }-\mathbb{E}[\Set{E}_{\SU }|\varepsilon_{\SM}]\big|\bar{f},\varepsilon_{\SM}\big]
\big)^{2t}\Big|\varepsilon_{\SM}\Big]
\label{Mc2t-F-E-c}\\
&\leq \bigg(\frac{3}{2}\bigg)^{2t} 
\mathbb{E}\Big[\mathbb{E}\big[\big(\Set{E}_{\SU }-\mathbb{E}[\Set{E}_{\SU }|\varepsilon_{\SM}]\big)^{2t}\big|\bar{f},\varepsilon_{\SM}\big]
\Big|\varepsilon_{\SM}\Big]
\label{Mc2t-F-E-d}\\
&= \bigg(\frac{3}{2}\bigg)^{2t} 
\mathbb{E}\big[\big(\Set{E}_{\SU }-\mathbb{E}[\Set{E}_{\SU }|\varepsilon_{\SM}]\big)^{2t}\big|\varepsilon_{\SM}\big]
\label{Mc2t-F-E-e}\\
&=\bigg(\frac{3}{2}\bigg)^{2t}\mathbb{M}^{(2t)}_{\mathrm{c}}\big[\Set{E}_{\SU }|\varepsilon_{\SM} \big],
\label{Mc2t-F-E-f}
\end{align}\label{Mc2t-F-E}
\end{subequations}
where \eqref{Mc2t-F-E-b} is obtained by substituting \eqref{eqn:ktoM-m} and \eqref{eqn:equalexp-m} into \eqref{Mc2t-F-E-a}, 
\eqref{Mc2t-F-E-d} is obtained by applying Jensen's inequality based on the fact that $x^{2t}$, $t\in \mathbb{Z}^+$, are convex functions,
and 
\eqref{Mc2t-F-E-e} is obtained by applying the tower property of conditional expectation.
In \eqref{Mc2t-F-E-f},  $\mathbb{M}^{(2t)}_{\mathrm{c}}\big[\Set{E}_{\SU }|\varepsilon_{\SM} \big]$ denotes the $2t$-th central moment of the posterior distribution of $\Set{E}_{\SU }$.

According to \eqref{Mc2t-F-E}, the even-order moments of the posterior distribution of $\mathcal{E}_{\SU }$ bound those of $\bar{F}$.
In this case, the application of Chebyshev's inequality (extended to higher moments) bounds the tail probability of $\bar{F}$.
The bounded tail probability of $\bar{F}$ leads to the credible interval ${\Set C}^{(T)}$, where $T$ indicates the maximum computed moment.

The formal propositions of this step are presented as \thref{lem:tailprob} and \thref{thm:cred_interval} in  Appendix~\ref{sec:tailprob}.

\subsection{Determine the maximum computed moment}
The above procedure yields the credible interval ${\Set C}^{(T)}$.
The following note presents a criterion for selecting the maximum computed moment $T$.

According to \thref{thm:cred_interval}, the credible interval ${\Set C}^{(T)}$ is legitimate in scenarios with general noise for all $T$ values that are even positive integers.
Selecting a larger $T$ increases the computational cost but can yield credible intervals with smaller widths and thus increase the estimation accuracy.
As shown in Fig.~\ref{fig_Cre_Int_width}, increasing $T$ significantly increases the accuracy of interval estimates.
In this context, it is desirable to find an optimal maximum computed moment $T^*$ that gives the most accurate estimate with a low computational cost.

We address the issue of identifying the optimal maximum computed moment $T^*$ by exploiting the fact that beta-binomial distributions approach normal distributions in the large sampling region \cite{Das:B10}.
The obtained expression for $T^*$ is given in \eqref{eqn:Tstar}. 
The detailed derivation can be found in  Appendix~\ref{sec:Tstar}.

\begin{figure}[t] \centering
\hspace*{-4mm}\includegraphics[scale=0.95]{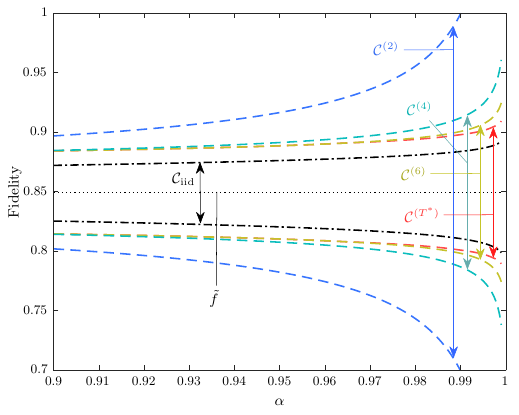}
\caption {Credible intervals assuming general noise with different maximum computed moments, and the credible interval assuming \ac{i.i.d.}~noise. In this figure, $N=10000$, $M=1000$, $\varepsilon_{\SM} =0.1$.
$\Set{C}_{\mathrm{iid}}$ is the credible interval obtained for the case with \ac{i.i.d.} noise, which serves as a lower bound on the width of all credible intervals. 
When the credible probability $\alpha = 0.99$, 
the credible interval obtained considering only the second moment, i.e., $\Set{C}^{(2)}$, is $208\%$ wider than the interval obtained considering all even moments, i.e., $\Set{C}^{(T^*)}$, demonstrating the need to bound the tail probability by higher-order moments.}
\label{fig_Cre_Int_width}
\end{figure}

\section{Impact of General Noise}
\label{sec:factors}
In this section, we analyze the effect of general noise on estimation accuracy, represented by the radius of the credible interval.

The expression for the credible interval $\Set{C}^{(T^*)}$ in Protocol~\ref{alg:fidelityest_int} is relatively complicated.
To illustrate the effects of general noise on estimation accuracy,  we analyze the expression of ${\Set C}^{(2)}$, the credible interval that bounds the tail probability of $\Set{P}(\bar{f}|\varepsilon_{\SM})$ only by the second-order moment. 
As a baseline for comparison, we also derive the credible interval in scenarios with \ac{i.i.d.} noise. 
The detailed derivation is in  Appendix~\ref{sec:iid}.

The radii of the credible intervals $\Set{C}^{(2)}$ and $\Set{C}_{\mathrm{iid}}$
are
\begin{subequations}
\begin{align}
\delta^{(2)}&=  \sqrt{\frac{1}{1-\alpha}}\sqrt{\frac{N+1}{N-M}}\sqrt{\frac{(2\tilde{f}+1)(1-\tilde{f})}{2(M+2)}}, \label{eqn:delta2}\\
\delta_{\mathrm{iid}} &= Q\Big(\frac{1+\alpha}{2}\Big)\sqrt{\frac{(2\tilde{f}+1)(1-\tilde{f})}{2(M+2)}}\label{eqn:delta-iid},
\end{align}\label{eqn:radii}
\end{subequations}
respectively. The following three factors determine the accuracy of the estimate.

\begin{itemize}[leftmargin=5mm]
\item[1.] Number of measurements.

The radii of the credible intervals, $\delta^{(2)}$ and $\delta_{\mathrm{iid}}$, have a common term, i.e., 
\begin{align}
\sqrt{\frac{(2\tilde{f}+1)(1-\tilde{f})}{2(M+2)}}.\label{eqn:factor1}
\end{align}
Recall the minimum mean squared estimation error characterized by the Cram\'er-Rao bound in \cite[Eq. 16]{Ruan:J23}.
In the case of \ac{i.i.d.} noise, \cite[Eq. 16]{Ruan:J23} shows that the minimum standard deviation of the estimated fidelity is given by
\begin{align}
\sqrt{\frac{(2\bar{f}+1)(1-\bar{f})}{2M}}.\label{eqn:errorbound-iid}
\end{align}
\eqref{eqn:factor1} and \eqref{eqn:errorbound-iid} approach the same value for large $M$.
Therefore, \eqref{eqn:factor1} represents the contribution from the $M$ number of measurements.

\item[2.] Tightness of tail probability bound.

The radii of the credible intervals, $\delta^{(2)}$ and $\delta_{\mathrm{iid}}$, contain a term associated with the credible probability $\alpha$, i.e.,
\begin{align}
\sqrt{\frac{1}{1-\alpha}}\label{eqn:sqrt-alpha}
\end{align}
in $\delta^{(2)}$ and 
\begin{align}
Q\Big(\frac{1+\alpha}{2}\Big)\label{eqn:Q-alpha}
\end{align}
in $\delta_{\mathrm{iid}}$.

Since $Q(\cdot)$ is the quantile function of the standard normal distribution, \eqref{eqn:Q-alpha} shows that in the case of \ac{i.i.d.} noise and a large number of measurements, the posterior distribution of $\bar{F}$ has an exponential tail similar to normal distributions.

In the case of general noise, the exact posterior distribution of $\bar{F}$ cannot be determined.
Instead, the even moments of the posterior distribution of $\mathcal{E}_{\SU }$ bound the tail probability of $\bar{F}$.
Only the second moment is used to obtain $\delta^{(2)}$. 
As indicated by \eqref{eqn:sqrt-alpha}, the bound on the tail probability decays quadratically in this case.
As shown in Fig.~\ref{fig_Cre_Int_width}, such a low decay rate leads to an inaccurate estimate, especially when the credible level $\alpha$ is close to $1$.
This finding highlights the importance of bounding the tail probability of $\bar{F}$ with higher-order moments.

\item[3.] Error due to atypical quantum states.

Compared to $\delta_{\mathrm{iid}}$, there is an additional term in $\delta^{(2)}$, i.e., 
\begin{align}
\sqrt{\frac{N+1}{N-M}}.\label{eqn:perturbation}
\end{align}
The value of \eqref{eqn:perturbation} is always greater than $1$ and increases monotonically as the number of measurements, $M$, approaches the total number of qubit pairs, $N$.
This term represents the additional estimation error introduced by the deviation between the average fidelity of all qubit pairs, $\bar{f}_{\mathrm{all}}$
and that of the unsampled ones, $\bar{f}$.

In scenarios with \ac{i.i.d.} noise, all qubit pairs have the same fidelity. 
However, in scenarios with general noise, $\bar{f}$ may deviate from $\bar{f}_{\mathrm{all}}$, leading to the additional estimation error factor expressed by \eqref{eqn:perturbation}.
In particular, as $M$ approaches $N$, the average fidelity of a smaller number of unsampled qubit pairs is likely to deviate more from that of all pairs, making \eqref{eqn:perturbation} a monotonically increasing function of $M$.
Fig.~\ref{fig_betabin} shows how the presence of a smaller number of unsampled qubit pairs leads to a less accurate estimate.
\end{itemize}

\begin{figure}[t] \centering
\vspace{-0.5mm}
\hspace*{-3mm}\includegraphics[scale=0.95]{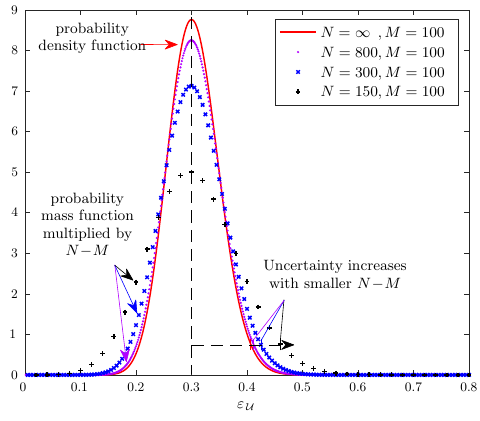}
\caption {Posterior distributions of the \ac{QBER} $\mathcal{E}_{\SU }$ of unsampled qubit pairs with different numbers of all qubit pairs, $N$.
The number of errors $e_{\SM}=30$ for all cases.
The posterior distribution of $\Set E_{\SU}$ becomes more dispersed with smaller $N-M$, indicating a higher degree of uncertainty in $\Set E_{\SU}$.
The tail probability of the average fidelity $\bar{F}$ is bounded by the posterior distribution of $\Set E_{\SU}$.
Consequently, the estimate of $\bar{F}$ becomes less accurate with smaller $N$.
}
\label{fig_betabin}
\end{figure}

Considering the three factors that affect the estimation accuracy, the following two remarks summarize the effects of general noise.

\begin{Remark}[Address general noise with high efficiency]
\cite{TakManMorMizFig:J19} addresses the issue of general noise for quantum state verification by using the Serfling’s bound.
When this approach is applied to measure entangled qubit pairs, the radius of the credible interval is given by \cite[Eq. 41]{TakManMorMizFig:J19}
\begin{align}
\frac{\sqrt{c}}{2} + 2\varepsilon, \label{eqn:Radi-Serf}
\end{align}
where $c$ is a fucntion of the credible probability $\alpha$, and $\varepsilon$ is the ratio that the measured states fail the test.

According to the central limit theorem, the radii of the credible intervals shall converge to $0$ on the order of $\mathcal{O}\Big(\sqrt{\frac{1}{M}}\Big)$ as the number of measurements $M$ approaches infinity in scenarios with \ac{i.i.d.} noise.
The radius $\delta_{\mathrm{iid}}$ in \eqref{eqn:delta-iid} follows this trend.
However, \eqref{eqn:Radi-Serf} remains a constant as  $M$ increases.
This fact highlights the problem of efficiency loss when dealing with general noise.
 
According to \eqref{eqn:delta2}, when the measurement ratio $\frac{M}{N}$ is a constant,
$\delta^{(2)}$ also converges to $0$ on the order of $\mathcal{O}\Big(\sqrt{\frac{1}{M}}\Big)$ as $M$ approaches infinity.
This fact shows that with a constant measurement ratio and increasing number of measurements, the estimation accuracy of the proposed protocol improves at the same rate as in the case with \ac{i.i.d.} noise.
Therefore, the proposed protocol addresses the issue of general noise with high efficiency.
~\hfill~\qed
\end{Remark}

\begin{Remark}[The issue of excessive measurements]
\label{remark:excessive}
When the total number of qubit pairs $N$ remains constant and the number of measurement $M$ increases, the first and third factors, i.e., \eqref{eqn:factor1} and \eqref{eqn:perturbation}, have conflicting effects on the estimation accuracy.
The information obtained from the measurements, \eqref{eqn:factor1}, makes $\delta^{(2)}$ smaller, while the error due to atypical unsampled qubit pairs, \eqref{eqn:perturbation}, has the opposite effect.

In particular, one can see that when
\begin{align}
M \geq \bigg\lfloor \frac{N}{2} \bigg\rfloor,\label{eqn:M-excessive}
\end{align}
$\delta^{(2)}$ is an increasing function of $M$, indicating an undesirable scenario in which more measurements lead to a less accurate estimate.
This issue of excessive measurements occurs because the error due to atypical unsampled qubit pairs becomes the dominant factor when the measurement ratio $\frac{M}{N}$ is large.
This issue is a unique challenge associated with general noise and should be avoided in practice.

\eqref{eqn:M-excessive} presents the criterion to avoid excessive measurements for the credible interval $\Set{C}^{(2)}$.
The proposed credible interval $\Set{C}^{(T^*)}$ uses multiple moments to bound the tail probability of $\bar{F}$.
Due to the complexity of the expressions of higher-order moments, it is challenging to find exact criteria to avoid excessive measurements for $\Set{C}^{(T^*)}$.
To overcome this challenge, it is worth noticing that the higher-order moments are also likely to be small when the second-order moment is minimized.
Therefore, it is reasonable to conjecture that \eqref{eqn:M-excessive} is also a satisfactory metric to avoid excessive measurements for $\Set{C}^{(T^*)}$.
To test this conjecture, a wide range of parameters, i.e., all parameter combinations satisfying
\begin{align*}
N&\in\{5000,10000,20000\},\\ 
\alpha &\in \{0.9, 0.95, 0.99\},\\
\varepsilon_{\SM}&\in\{0,\frac{50}{N},2\cdot\frac{50}{N},\ldots,0.66\},
\end{align*} 
are tested \footnote{\label{ft:nohighBER} Becase $\mathbb{E}\big[{\Set E}_{\SU }\big] = \frac{2}{3}$ when $\bar{f}=0$, cases with $\varepsilon_{\SM}>\frac{2}{3}$ are unlikely to occur in practice. Hence, such cases are not considered.}.
These tests show that \eqref{eqn:M-excessive} characterizes the boundary of excessive measurements for $\Set{C}^{(T^*)}$ as long as the \ac{QBER} $\varepsilon_{\SM}$ satisfies
\begin{align}
\varepsilon_{\SM} \geq 0.01.
\end{align} 
Fig.~\ref{fig_Cre_Int_minM} shows the test results of a few representative parameter combinations.

Practical quantum systems rarely yield a \ac{QBER} $\varepsilon_{\SM}<0.01$. 
Therefore, \eqref{eqn:M-excessive} is an accurate metric to avoid excessive measurements in practice.
This aspect means that without prior information of the noise, the measurement ratio $\frac{M}{N}$ shall not exceed $0.5$.
~\hfill~\qed
\end{Remark}

\begin{figure}[t] \centering
\includegraphics[scale=0.95]{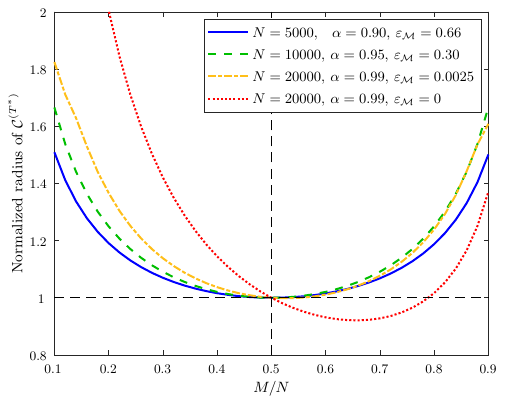}
\caption {Radius of the credible interval $\Set{C}^{(T^*)}$ as a function of the measurement ratio ${M}/{N}$ for different parameter combinations. To facilitate the comparison, the radius of $\Set{C}^{(T^*)}$ is normalized, i.e., divided by the radius at the point ${M}/{N}=0.5$.
}
\label{fig_Cre_Int_minM}
\end{figure}

\section{Protocol Evaluation}
\label{sec:sim}
In this section, we evaluate the proposed protocol using numerical tests.

\noindent {\bf Noise model:}
\label{subsec:perf_eva}
The noise is modeled as heterogeneous and correlated depolarizing channels, and the corresponding density matrix of all qubit pairs is
\begin{align}
\V{\rho}_{\mathrm{all}} = \frac{1}{2}\sum_{x\in \{h,l\}}\big(\otimes^{N q^{(x)}} \V{\rho}^{(g)}\big) \otimes \big(\otimes^{N (1-q^{(x)})} \V{\rho}^{(b)}\big), 
\end{align}
where the state of each qubit pair
\begin{align}
\V{\rho}^{(s)} =  p^{(s)}\frac{\mathbb{I}_4}{4} + (1-p^{(s)})|\Psi^-\rangle\langle\Psi^-|, \quad s\in\{g,b\}
\end{align}
with the error probabilities $p^{(s)}$, $s\in\{g,b\}$ of good and bad channels satisfy $0\leq p^{(g)} \leq  p^{(b)} \leq 1$,
and the rate of good channels $q^{(x)}\in[0,1]$, $x\in\{h,l\}$, is set such that $Nq^{(x)}$ are integers.~\QEDB

The difference between the error probabilities of the bad and the good channels, i.e., 
\begin{align}
d= p^{(b)}-p^{(g)}\in[0,1],
\end{align} 
represents the degree of heterogeneity and correlation in the noise.
In particular,  the noise is \ac{i.i.d.}~when $d=0$.

The representative test computes the credible intervals assuming \ac{i.i.d.}~and general noise, i.e., $\Set{C}_{\mathrm{iid}}$ and $\Set{C}^{(T^*)}$  as functions of the number of measurements, $M$.
To generate fixed credible intervals for each $M$, this test considers the special case in which the realized \ac{QBER} $\varepsilon_{\SM}$ is equal to its expected value.
Fig.~\ref{fig_True_VS_Int_M} shows the plot of credible intervals.

\begin{figure}[t] \centering
\hspace*{-0.4cm}\includegraphics[scale=0.95]{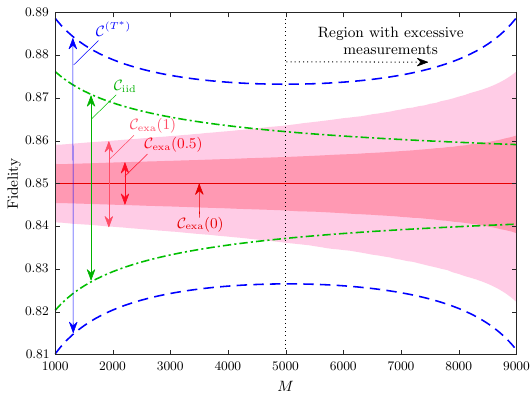}
\caption {The credible intervals in the case of \ac{i.i.d.}~or general noise, i.e., $\Set{C}_{\mathrm{iid}}$ and $\Set{C}^{(T^*)}$, are plotted as a function of the number of measurements, $M$, under the condition that the expected \ac{QBER} was observed.
The exact percentile interval $\Set{C}_{\mathrm{exa}}(d)$, i.e., the interval between the percentiles $100\frac{\alpha}{2}$ and $100\big(1-\frac{\alpha}{2}\big)$ of the posterior distribution of fidelity $\bar{F}$, is also plotted, where $d\in\{0,0.5,1\}$ represents the degree of heterogeneity and correlation in the noise.
$N=10000$, $\alpha = 0.95$, $r^{(h)} = 0.81$, and $r^{(l)} = 0.79$. }
\label{fig_True_VS_Int_M}
\end{figure}

As benchmarks, the exact percentiel intervals $\Set{C}_{\mathrm{exa}}(d)$ of the average fidelity $\bar{F}$ are also plotted,
where $d\in\{0,0.5,1\}$ represents the degree of heterogeneity and correlation in the noise.
Fig.~\ref{fig_True_VS_Int_M} shows that, with the exception of $\Set{C}_{\mathrm{exa}}(0)$, the widths of the exact intervals $\Set{C}_{\mathrm{exa}}(0.5)$ and $\Set{C}_{\mathrm{exa}}(1)$ increase with increasing $M$.
This fact is attributed to the error introduced by atypical unsampled qubit pairs.
In scenarios with non-\ac{i.i.d.} noise, the average fidelity $\bar{F}$ of the unsampled qubit pairs likely deviates from that of all pairs.
A smaller number of unsampled qubit pairs, $M-N$, and a higher degree of heterogeneity and correlation, $d$, 
correspond to a more significant deviation, making the distribution of $\bar{F}$ more scattered.

To be valid credible intervals, $\Set{C}_{\mathrm{iid}}$ and $\Set{C}^{(T^*)}$ must cover intervals $\Set{C}_{\mathrm{exa}}(d)$, $d\in\{0,0.5,1\}$.
In this context, $\Set{C}_{\mathrm{iid}}$ is no longer valid when the number of measurements, $M$, and the degree of heterogeneity, $d$, increase above certain thresholds.
In contrast, $\Set{C}^{(T^*)}$ is valid for all values of $M$ and $d$, proving its reliability.
This result highlights the need to use the proposed estimation protocol in scenarios with non-\ac{i.i.d.} noise.

Moreover, $\Set{C}^{(T^*)}$ is narrowest when the measurement ratio $\frac{M}{N}=0.5$, which is consistent with Remark~\ref{remark:excessive}.

In addition to this representative test, other tests are performed to investigate the performance of the proposed protocol in different scenarios. Fig.~\ref{fig_True_VS_Int_Joint} shows the plot of the correct probability of the credible intervals, i.e., the probability that the true fidelity $\bar{F}$ lies within the credible intervals, $\Set{C}^{(T^*)}$ or $\Set{C}_{\mathrm{iid}}$, averaged over all realizations.

Fig.~\ref{fig_True_VS_Int_Joint}A shows the plot of the correct probabilities as a function of the number of measured qubit pairs $M$. 
The correct probability of $\Set{C}_{\mathrm{iid}}$ decreases significantly with increasing $M$, which is consistent with the results shown in Fig.~\ref{fig_True_VS_Int_M}.

Fig.~\ref{fig_True_VS_Int_Joint}B shows the plot of correct probabilities as a function of the degree of heterogeneity and correlation $d$. 
When $d=0$, the noise is \ac{i.i.d.}, so the correct probability of $\Set{C}_{\mathrm{iid}}$ is equal to the credible probability $\alpha=0.95$.
However, this probability decreases significantly with increasing $d$ due to higher heterogeneity and correlation level.

Fig.~\ref{fig_True_VS_Int_Joint}C shows the plot of correct probabilities as a function of the credible probability $\alpha$. 
The average correct probability of the credible interval assuming \ac{i.i.d.}~noise increases with increasing $\alpha$, but is $0.04$ to $0.08$ lower than $\alpha$.

In all cases considered, the correct probability of the proposed credible interval $\Set{C}^{(T^*)}$ is higher than the credible probability $\alpha$, confirming the reliability of Protocol~\ref{alg:fidelityest_int} in the presence of general noise.

\begin{figure}[t] \centering
\includegraphics[scale=1.02]{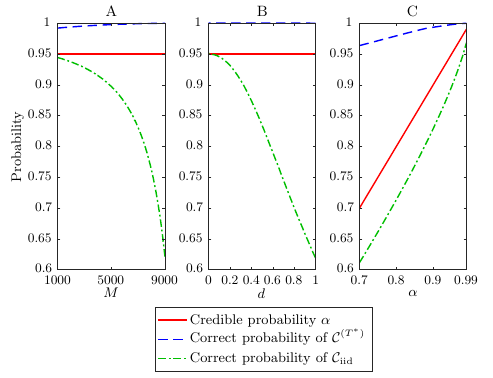}
\caption {Correct probabilities, i.e., the probabilities that the true fidelity lies in the credible intervals over all realizations. 
The parameter settings in the cases shown in the three subfigures are the same as in Fig.~\ref{fig_True_VS_Int_M}, 
except that in subfigure A, the degree of heterogeneity and correlation in the noise $d =1$;
in subfigure B, the number of measurements $M=9000$, and  $d$ is between $0$ and $1$; and
in subfigure C, $M=5000$, $d =1$, and the credible probability $\alpha$ is between $0.7$ and $0.99$.
}
\label{fig_True_VS_Int_Joint}
\end{figure}

\section{Conclusion}
\label{sec:conclusion}
The effect of non-\ac{i.i.d.} noise is a major challenge for fidelity estimation of entangled states in quantum networks.
In this paper, by performing random sampling, a thought experiment, and Bayesian inference,  we obtain a credible interval of fidelity that is reliable in scenarios with general noise.

Further analysis characterizes the factors that determine the accuracy of fidelity estimation.
In particular, the error due to atypical unsampled qubit pair leads to the issue of excessive measurements, a distinctive phenomenon that is not present in scenarios with \ac{i.i.d.} noise.
Numerical tests and further analysis have confirmed the reliability and accuracy of the proposed estimation protocols.

\appendix
\section*{Notations}
In this document, random variables and their realizations are represented in upper and lower case letters, e.g., $F$ and $f$. 
Vectors and matrices are denoted by bold letters, e.g., $\V{\rho}$. 
$\mathbb{E}[\cdot]$, $\mathbb{V}[\cdot]$, $\mathbb{M}^{(j)}[\cdot]$, and $\mathbb{M}_{\mathrm{c}}^{(j)}[\cdot]$ denote, respectively, the expectation, variance, $j$-th raw moment and $j$-th central moment of a random variable.
The statistics of $A$ conditional on $X=x$ is denoted as
$\mathbb{S}[A|X=x]$, or $\mathbb{S}[A|x]$, for brevity, where $\mathbb{S}\in\{\mathbb{E}, \mathbb{V}, \mathbb{M}^{(j)}, \mathbb{M}_{\mathrm{c}}^{(j)}\}$.
$\mathbb{I}_n$ denotes an $n\times n$ identity matrix.
The function $\Set{P}(x)$ denotes the distribution of the random variable $X$, i.e. \ac{pdf} or \ac{pmf}.
$\mathrm{Pr}(\cdot)$ denotes the probability of an event,
$\mathrm{Tr}(\cdot)$ denotes the trace of a matrix, and $\mathrm{Tr}_{n}(\cdot)$ denotes the partial trace \ac{w.r.t.} the $n$-th subsystem, 
$\Gamma(\cdot)$ denotes the gamma function,
and
$\mathcal{B}(a,b)$ denotes the beta function with parameters $a$ and $b$.
The functions $\lfloor\cdot \rfloor$ and $\lfloor\cdot \rceil$ denote the floor and rounding (nearest integer) functions, respectively.
\label{sec:symbol}

\section{Frequentism and Bayesianism}
\label{sec:method}
In statistical analysis, there are two main paradigms: frequentism and Bayesianism \cite{Efr:J05}.
As summarized  in Table~\ref{tab:dual}, there is a duality between the two paradigms: each key notion in one paradigm has a counterpart in the other paradigm with a similar definition.

\begin{table*}[t]
\renewcommand{\arraystretch}{1.2}
\caption{Key notions in Frequentism and Bayesianism$^*$\label{tab:dual}}
\centerline{
\begin{threeparttable}
\begin{tabular}{c|c"c|c|c|c}
\hlinew{1pt}
\multicolumn{2}{c"}{} & Point estimator & Estimation error & Strength of evidence  & Interval estimator   \\
\hline
\multirow{4}{*}{\hspace{-1mm}{\parbox[c]{2.3cm}{\centering Frequentism}}\hspace{-1mm}}
&\Clb &
\Clb Unbiased &\Clb Variance of &\Clb  &\Clb Confidence\\
&\Clb  \multirow{-2}{*}{\hspace{-1mm}{\parbox[c]{2.3cm}{\centering Terminology}}\hspace{-1mm}}& 
\Clb estimator&\Clb the estimator&\Clb \multirow{-2}{*}{\hspace{-1mm}{\parbox[c]{2.3cm}{\centering $p$-value}}\hspace{-1mm}}                &\Clb interval\\\cline{2-6}
& \multirow{2}{*}{\hspace{-1mm}{\parbox[c]{2.3cm}{\centering Definition}}\hspace{-1mm}}
& $\tilde{\theta}(x)\mbox{ such that}$ 
& \multirow{2}{*}{\hspace{-1mm}{\parbox[c]{2.7cm}{\centering $\mathbb{E}\big[(\tilde{\theta}(X)-\theta)^2 \big|\theta \big]$}}\hspace{-1mm}} &
\multirow{2}{*}{\hspace{-1mm}{\parbox[c]{3.5cm}{\centering $\Pr\big[X\in{\Set S}(x) \big|H(\theta)\big]$}}\hspace{-1mm}}
& ${\Set C}(x)$ such that \\
&& $\mathbb{E}\big[\tilde{\theta}(X) \big|\theta \big]=\theta$ &  
& & $\Pr\big[\theta\in  {\Set C}(X) \big|\theta\big]\geq \alpha$ \\\hline
\multirow{4}{*}{\hspace{-1mm}{\parbox[c]{2.3cm}{\centering Bayesianism}}\hspace{-1mm}}     
&\Clb &
\Clb Bayes &\Clb Mean square &\Clb  Posterior probability &\Clb Credible \\
&\Clb \multirow{-2}{*}{\hspace{-1mm}{\parbox[c]{2.3cm}{\centering Terminology}}\hspace{-1mm}} &
\Clb estimator  &\Clb error &\Clb  of hypothesis &\Clb interval \\\cline{2-6}
&\multirow{2}{*}{\hspace{-1mm}{\parbox[c]{2.3cm}{\centering Definition}}\hspace{-1mm}} & $\tilde{\theta}(x)\mbox{ such that}$ 
& \multirow{2}{*}{\hspace{0mm}{\parbox[c]{2.9cm}{\centering $\mathbb{E}\big[(\tilde{\theta}(x)-\Theta)^2 \big| x \big]$}}\hspace{-0.5mm}}
& \multirow{2}{*}{\hspace{-1mm}{\parbox[c]{3cm}{\centering $\Pr\big[ H(\Theta)\big| x\big]$ }}\hspace{-1mm}}
 &  ${\Set C}(x)$ such that\\
& & $\mathbb{E}\big[\Theta \big| x \big]=\tilde{\theta}(x)$ & &  
& $\Pr\big[\Theta\in {\Set C}(x) \big| x\big]\geq \alpha$ \\
\hlinew{1pt}
\end{tabular}
\begin{tablenotes}\footnotesize
\item[*] $\theta$ is the parameter to be estimated, $x$ is the observed data point, $\tilde{\theta}(x)$ is an estimator of $\theta$ based on $x$, $H(\theta)$ is a hypothesis of $\theta$, and ${\Set S}(x)$ and ${\Set C}(x)$ are two intervals determined by $x$. Variables are indicated in upper or lower case letters depending on whether they are regarded as random or deterministic in the context.
\end{tablenotes}
\end{threeparttable}
}
\end{table*}

The main difference between the two paradigms is that the notions in frequentism are defined conditioned on the parameter $\theta$, whereas those in Bayesianism are defined conditioned on the observed data $x$.
Specifically, frequentism assumes that $\theta$ is a fixed unknown parameter and improves the estimation accuracy by performing repeated experiments conditioned on the same parameter, while
Bayesianism treats $\theta$ as a random variable and refines the distribution of $\theta$ conditioned on the observed data $x$
\cite{BerWol:B88,BerSmi:B00,Wagetal:J08}.

For fidelity estimation in scenarios with non-\ac{i.i.d.} noise, the estimation target, i.e., the average fidelity of the unsampled qubit pairs $\bar{f}$, changes as a function of both the sampling set $\Set M$ and the measurement outcome $\V{r}$.
Consequently, measurements on different qubit pairs are not repeated experiments conditioned on the same parameter.
In this respect, the Bayesian paradigm is more suitable for making estimates for fidelity estimation in scenarios with non-\ac{i.i.d.} noise. 
Therefore, the credible interval, instead of the confidence interval, is used as the interval estimate.

\section{Derivation of the Credible Interval}
\label{sec:Bayesian}
Consider a thought experiment that measures all qubit pairs.
The goal is to introduce the \ac{QBER} of unsampled qubit pairs $\mathcal{E}_{\SU }$ and use it to establish a relationship between the measured \ac{QBER} $\varepsilon_{\SM}$ and the tail probability of the average fidelity $\bar{F}$ of the unsampled qubit pairs.
Accordingly, the interval estimation consists of two main steps, i.e., characterizing the moments of the posterior distribution of $\mathcal{E}_{\SU }$,
and bounding the tail probability of $\bar{F}$ based on the moments characterized in the previous step.

\subsection{Posterior distribution of $\mathcal{E}_{\SU }$}
\label{pf_lem:unsampled}
In this subsection, we derive the moments of the posterior distribution of  $\mathcal{E}_{\SU }$, the \ac{QBER} of the unsampled qubit pairs. 

Let us first consider the following sampling problem:
$b_n, n\in \Set N$ are binary variables, with $|\Set{N}|=N$, and there is no prior information about the sum of all variables $s_{\Set N} =\sum_{n\in \Set N}b_n$.
The sample set $\Set{M} \subseteq \Set N$, $|\Set{M}|=M$ is drawn completely at random.
Denote the sum of the sampled and unsampled variables as
\begin{align}
&S_{\SM}=\sum_{n\in\SM} b_n,& S_{\SU}=\sum_{n\in\SU}b_n,
\end{align}
respectively, where the unsampled set $\SU = \Set N \backslash \SM$.
The following lemma characterizes the posterior distribution of $S_{\SU }$.

\begin{LemA}[Posterior distribution of the unsampled variables]\thlabel{lem:unsampled}
Given a realization of the sample result $S_{\SM}=s_{\SM}$,
the posterior distribution of $S_{\SU}$ is given by
\begin{align}
&{\Set P}(s_{\SU }| s_{\SM})\nonumber\\
&={N-M\choose s_{\SU }}\frac{\mathcal{B}(s_{\SU }+s_{\SM}+\frac{1}{2},N-s_{\SU }- s_{\SM}+\frac{1}{2})}{\mathcal{B}(s_{\SM}+\frac{1}{2},M- s_{\SM}+\frac{1}{2})},\label{eqn:pKbar}
\end{align}
where $\mathcal{B}(a,b)$ is the beta function with parameters $a$ and $b$.
\end{LemA}

\begin{proof}
There is no prior information about the sum of the variables, $s_{\Set N}$.
In the Bayesian paradigm, when no prior information is available, a noninformative prior must be selected to reflect a balance among possible parameter values.
Several types of noninformative priors have been proposed \cite{BerBerSun:J09,StaRei:J10,Kas:J90}.
The inference of the sum of the unsampled variables is a one-dimensional problem, in which the noninformative priors proposed in \cite{BerBerSun:J09,StaRei:J10,Kas:J90} all converge to the Jeffreys prior \cite{Syv:98}. 
Therefore, the Jeffreys prior is considered to be the optimal noninformative prior for one-dimensional inference problems \cite{BerBerJosDon:J15}.
Following the literature, the Jeffreys prior is adopted.

Since $S_{\Set N}$ is discrete, its noninformative prior cannot be directly obtained via the Jeffreys prior.
To solve this problem, we adopt a hierarchical setting \cite[p.339]{BerSmi:B00} to generate the noninformative prior of $S_{\Set N}$.
Specifically, $S_{\Set N}$ is drawn from a binomial distribution with mean $\Theta$ and $N$ number of trials, where the distribution of $\Theta$ is given by the Jeffreys prior, i.e., 
\begin{align}
{\Set P}(\theta)=\frac{\theta^{-\frac{1}{2}}(1-\theta)^{-\frac{1}{2}}}{\mathcal{B}(\frac{1}{2},\frac{1}{2})}.
\end{align}
With this hierarchical setting, the noninformative prior of $S_{\Set N}$ is specified by \cite[eqn.(8)]{BerBerSun:J12}:
\begin{subequations}
\begin{align}
{\Set P}(s_{\Set N}) &= \int_0^1{\Set P}(s_{\Set N}|\theta){\Set P}(\theta)d\theta\nonumber\\
&=\int_0^1{N\choose s_{\Set N}}
\theta^{s_{\Set N}}(1-\theta)^{N- s_{\Set N}}
\frac{\theta^{-\frac{1}{2}}(1-\theta)^{-\frac{1}{2}}}{\mathcal{B}(\frac{1}{2},\frac{1}{2})} d\theta\nonumber\\
&=\frac{1}{\pi}{N\choose s_{\Set N}}
 \int_0^1 \theta^{s_{\Set N}-\frac{1}{2}}(1-\theta)^{N- s_{\Set N}-\frac{1}{2}} d\theta
 \label{eqn:pKM_prior_2a}
 \\
&=\frac{1}{\pi}\frac{\Gamma(s_{\Set N} +\frac{1}{2})\Gamma(N-s_{\Set N}+\frac{1}{2})}{\Gamma(s_{\Set N} +1)\Gamma(N-s_{\Set N}+1)},\label{eqn:pKM_prior_2b}
\end{align}
\end{subequations}
where $\Gamma(\cdot)$ is the gamma function, and \eqref{eqn:pKM_prior_2b} is obtained by substituting the following two equations.
\begin{align}
{x\choose y}=\frac{\Gamma(x+1)}{\Gamma(y+1)\Gamma(x-y+1)},\quad \forall x, y\in\mathbb{C},  \label{eqn:comb_gamma}
\end{align}
and
\begin{align}
&\mathcal{B}(x,y) = \int_0^1 \theta^{x-1}(1-\theta)^{y-1} d\theta
=\frac{\Gamma(x)\Gamma(y)}{\Gamma(x+y)},\nonumber\\
& \forall x,y >0.
  \label{eqn:beta_gamma}
\end{align}

Denote the number of unsampled variables by $U = N-M$. 
Given the sum of all variables, $S_{\Set N}$, the sum of sampled variables $S_{\SM}$ follows a hypergeometric distribution.
Therefore, from \eqref{eqn:pKM_prior_2a}, the prior distribution of $S_{\SM}$ can be derived as follows.
\begin{subequations}
\begin{align}
&{\Set P}(s_{\SM})\nonumber\\
&= \sum_{s_{\Set N} =0}^{N} {\Set P}(s_{\SM}|s_{\Set N}){\Set P}(s_{\Set N})\nonumber\\
&=  \sum_{s_{\Set N} =0}^{N}
\frac{{s_{\Set N}\choose s_{\SM}}
{N-s_{\Set N}\choose M-s_{\SM}} }
{ {N\choose M} }\frac{1}{\pi}{N\choose s_{\Set N} }\nonumber
 \\&
 \hspace{4.5mm}\int_0^1 \theta^{s_{\Set N}-\frac{1}{2}}(1-\theta)^{N- s_{\Set N}-\frac{1}{2}} d\theta\nonumber\\
 &=  \frac{1}{\pi} \int_0^1\sum_{s_{\Set N} =s_{\SM}}^{N-M+s_{\SM}}
\frac{{s_{\Set N}\choose s_{\SM}}
{N-s_{\Set N}\choose M-s_{\SM} }
{N\choose s_{\Set N}}}
{{N\choose M}}\nonumber\\ 
&\hspace{4.5mm}\cdot
 \theta^{s_{\Set N}-\frac{1}{2}}(1-\theta)^{N- s_{\Set N}-\frac{1}{2}} d\theta \label{eqn:psm_a}
 \end{align}
 \begin{align}
 &= \frac{1}{\pi} \int_0^1\sum_{s_{\Set N} =s_{\SM}}^{N-M+s_{\SM}}
{M\choose s_{\SM}}
{N-M\choose s_{\Set N}- s_{\SM}}
\nonumber\\ 
&\hspace{4.5mm}\cdot
 \theta^{s_{\Set N}-\frac{1}{2}}(1-\theta)^{N- s_{\Set N}-\frac{1}{2}} d\theta \label{eqn:psm_b}\\
 &= \frac{1}{\pi}{M\choose s_{\SM}} \int_0^1
\theta^{s_{\SM}-\frac{1}{2}}(1-\theta)^{M- s_{\SM}-\frac{1}{2}}
\nonumber\\ 
&\hspace{4.5mm}\cdot
\sum_{s_{\SU} =0}^{U}
{U\choose s_{\SU}}\,
 \theta^{s_{\SU}}(1-\theta)^{U -s_{\SU}} d\theta  \label{eqn:psm_c} \\
 &= \frac{1}{\pi} {M\choose s_{\SM}}\int_0^1
\theta^{s_{\SM}-\frac{1}{2}}(1-\theta)^{M- s_{\SM}-\frac{1}{2}}d\theta  \label{eqn:psm_d}\\
&=\frac{1}{\pi}
\frac{\Gamma(s_{\SM}+\frac{1}{2})\Gamma(M- s_{\SM}+\frac{1}{2})}
{\Gamma(s_{\SM}+1)\Gamma(M- s_{\SM}+1)},\label{eqn:psm_e}
\end{align}
\label{eqn:psm}
\end{subequations}
where
\eqref{eqn:psm_a} is obtained by applying Fubini's theorem and the fact that ${A \choose B} =0$ if $B>A$,
\eqref{eqn:psm_b} is obtained by applying \eqref{eqn:comb_gamma},
\eqref{eqn:psm_d} holds because that $\theta^{s_{\SU}}(1-\theta)^{U -s_{\SU}}$ is the \ac{pmf} of a binomial distribution with mean $\theta$ and $U$ trials, therefore summing over all possible values of $s_{\SU}$ equals $1$, i.e.,
\begin{align}
\sum_{s_{\SU} =0}^{U}
{U\choose s_{\SU}}\,
 \theta^{s_{\SU}}(1-\theta)^{U -s_{\SU}}=1,\label{eqn:sumeq1}
 \end{align}
and \eqref{eqn:psm_e} is obtained by applying  \eqref{eqn:comb_gamma} and \eqref{eqn:beta_gamma}.

Using \eqref{eqn:pKM_prior_2b} and \eqref{eqn:psm}, the posterior distribution of $S_{\Set N}$ is derived by applying Bayes' theorem.
\begin{subequations} 
\begin{align}
&{\Set P}(s_{\Set N}|s_{\SM})\nonumber\\
&=\frac{{\Set P}(s_{\SM}|s_{\Set N}){\Set P}(s_{\Set N})}{{\Set P}(s_{\SM})}\nonumber\\
&=\frac{{s_{\Set N}\choose s_{\SM}}
{N-s_{\Set N}\choose M-s_{\SM}} }
{ {N\choose M} }\frac{\Gamma(s_{\Set N} +\frac{1}{2})\Gamma(N-s_{\Set N}+\frac{1}{2})}{\Gamma(s_{\Set N} +1)\Gamma(N-s_{\Set N}+1)}\nonumber\\
&\hspace{4.5mm}\cdot \frac{\Gamma(s_{\SM}+1)\Gamma(M- s_{\SM}+1)}
{\Gamma(s_{\SM}+\frac{1}{2})\Gamma(M- s_{\SM}+\frac{1}{2})}\nonumber\\
&=
{ N-M\choose s_{\Set N} - s_{\SM} }
\frac{\Gamma(s_{\Set N} +\frac{1}{2})\Gamma(N-s_{\Set N} +\frac{1}{2})}{\Gamma(N+1)}\nonumber\\
&\hspace{4.5mm}\cdot
\frac{\Gamma(M+1)}{\Gamma(s_{\SM} +\frac{1}{2})\Gamma(M-s_{\SM} +\frac{1}{2})}\label{eqn:pSN_SM_a}\\
&=
{ N-M\choose s_{\Set N} - s_{\SM} }
\frac{\mathcal{B}(s_{\Set N}+\frac{1}{2},N- s_{\Set N}+\frac{1}{2})}{\mathcal{B}(s_{\SM}+\frac{1}{2},M- s_{\SM}+\frac{1}{2})},
\label{eqn:pSN_SM_b}
\end{align} \label{eqn:pSN_SM}
\end{subequations} 
 where \eqref{eqn:pSN_SM_a}  is obtained by applying \eqref{eqn:comb_gamma}, and \eqref{eqn:pSN_SM_b} is obtained by applying \eqref{eqn:beta_gamma}.
 
According to \eqref{eqn:pSN_SM}, the posterior distribution of the sum of unsampled variables $S_{\SU }=S_{\Set N} - S_{\SM}$ is given by \eqref{eqn:pKbar}.
This completes the proof of \thref{lem:unsampled}.
 \end{proof}

The following lemma specifies the moments of the \ac{QBER} of unsampled qubit pairs, ${\Set E}_{\SU }$, based on the thought experiment and \thref{lem:unsampled}.
 
 \begin{LemA}[Moments of the posterior distribution of ${\Set E}_{\SU }$]\thlabel{lem:moments}
Given a realization of the \ac{QBER} of sampled qubit pairs, ${\Set E}_{\SM}=\varepsilon_{\SM}$,
the expected value of the \ac{QBER} of unsampled qubit pairs, ${\Set E}_{\SU }$, can be expressed as
\begin{align}
\varepsilon_{\mathrm E}=\mathbb{E}\big[{\Set E}_{\SU }\big| \varepsilon_{\SM}\big]=\varepsilon_{\SM} + \frac{\frac{1}{2}-\varepsilon_{\SM}}{M+1}. \label{eqn:kmexp_ap}
\end{align}
Moreover, the $t$-th central moment of the posterior distribution of ${\Set E}_{\SU }$, $t\in \mathbb{Z}^+$, is given by
\begin{align}
\mathbb{M}_{\mathrm{c}}^{(t)}\big[{\Set E}_{\SU }\big| \varepsilon_{\SM}\big]= 
\sum_{j=0}^{t}\frac{(-1)^j}{(N-M)^{(t-j)}}\bigg(\begin{array}{c}t\\j\end{array}\bigg)\big(\varepsilon_{\mathrm E}\big)^j\nonumber\\
\cdot\mathbb{M}^{(t-j)}\big[S_{\SU }\big| s_{\SM}\big],\label{eqn:MCkbar_ap}
\end{align}
where $S_{\SM}$ and $S_{\SU}$ denote the sum of the measurement outcome $r_n$ of the sampled $n\in \SM$ and unsampled $n\in \Set \SU$ qubit pairs, respectively,
$\mathbb{M}^{(j)}\big[S_{\SU }\big| s_{\SM}\big]$ is the $j$-th raw moment of the posterior distribution of $S_{\SU }$, given by
\begin{align}
&\mathbb{M}^{(j)}\big[S_{\SU }\big| s_{\SM}\big]=
\sum_{l=0}^{j}\bigg\{\begin{array}{c}j\\l\end{array}\bigg\}(N-M)_{(l)}\nonumber\\
&\hspace{15mm}\cdot\frac{\mathcal{B}\big(l+(M+1)\varepsilon_{\mathrm E},(M+1)(1-\varepsilon_{\mathrm E})\big)}{\mathcal{B}\big((M+1)\varepsilon_{\mathrm E},(M+1)(1-\varepsilon_{\mathrm E})\big)},\label{eqn:MR_ap}
\end{align}
in which $\bigg\{\begin{array}{c}j\\l\end{array}\bigg\}$ and $n_{(l)}$ denote the Stirling numbers of the second kind and the descending factorial function, respectively.
In particular, the variance (second central moment) of $\mathcal{E}_{\SU }$ is given by
\begin{align}
\mathbb{V}\big[{\Set E}_{\SU } | \varepsilon_{\SM}\big]&=\frac{N+1}{(N-M)(M+2)}\varepsilon_{\mathrm E}(1-\varepsilon_{\mathrm E}).\label{eqn:kmvar_ap}
\end{align}
\end{LemA}
\begin{proof}
The measurements on different qubit pairs are commutative. 
Thus, the order in which the measurements are made does not change the results of the measurements.
Therefore, one can perform the thought experiment in the following order:
\begin{enumerate}[leftmargin=*]
\item \emph{Measure all qubit pairs:} The nodes measure all qubit pairs and record the measurement outcomes $r_n\in\{0,1\}$, $n\in\Set{N}$.
\item \emph{Random sampling:} The nodes choose a subset of qubit pairs, $\SM$, and compute the \ac{QBER} of the sampled pairs, i.e., $\varepsilon_{\SM}= \frac{\sum_{n\in \SM} r_{n}}{M}$.
\item \emph{Post processing:} Given $\varepsilon_{\SM}$, the nodes derive the posterior distribution $\Set{P}(\varepsilon_{\SU }|\varepsilon_{\SM})$ and characterize its statistical properties.
\end{enumerate}
In this order, the estimation process after step 1 deals with a set of classical binary variables, i.e., $r_n\in\{0,1\}$, $n\in\Set{N}$.
Since there is no prior information about the state of the qubit pairs, there is no prior information about the sum of the measurement outcome, $s_{\Set N} = \sum_{n\in \Set N} r_n$.
Moreover, the sample set $\SM$ is drawn completely at random.
With all the preconditions of Lemma~\ref{lem:unsampled} satisfied, 
one can apply this lemma and obtain the posterior distribution of $S_{\SU }$:
\begin{align}
&{\Set P}(s_{\SU }| s_{\SM})\nonumber\\
&={N-M\choose s_{\SU }}\frac{\mathcal{B}(s_{\SU }+s_{\SM}+\frac{1}{2},N-s_{\SU }- s_{\SM}+\frac{1}{2})}{\mathcal{B}(s_{\SM}+\frac{1}{2},M- s_{\SM}+\frac{1}{2})}.\label{eqn:pKbar-2}
\end{align}

According to \eqref{eqn:pKbar-2}, the posterior distribution of $S_{\SU }$ is a beta-binomial distribution with parameters $N-M$, 
$s_{\SM}+\frac{1}{2}$, and $M- s_{\SM}+\frac{1}{2}$. 
Hence, according to \cite[eqn.(9)]{Oga:J16}, the $j$-th moment of the posterior distribution of $S_{\SU }$ is given by
\begin{align}
&\mathbb{M}^{(j)}\big[S_{\SU }\big| s_{\SM}\big]\nonumber\\
&=
\sum_{l=0}^{j}\bigg\{\begin{array}{c}j\\l\end{array}\bigg\}(N-M)_{(l)}
\frac{\mathcal{B}\big(l+s_{\SM}+\frac{1}{2}, M- s_{\SM}+\frac{1}{2}\big)}{\mathcal{B}\big(s_{\SM}+\frac{1}{2}, M- s_{\SM}+\frac{1}{2}\big)}
\nonumber\\
&=\sum_{l=0}^{j}\bigg\{\begin{array}{c}j\\l\end{array}\bigg\}(N-M)_{(l)}\nonumber
\\&\quad\cdot\frac{\mathcal{B}\big(l+(M+1)\varepsilon_{\mathrm E},(M+1)(1-\varepsilon_{\mathrm E})\big)}{\mathcal{B}\big((M+1)\varepsilon_{\mathrm E},(M+1)(1-\varepsilon_{\mathrm E})\big)},\label{eqn:Mraw_2}
\end{align}
which proves \eqref{eqn:MR_ap}.

In \eqref{eqn:Mraw_2}, by letting $j=1$, the posterior expectation of $S_{\SU }$ can be obtained as follows.
\begin{align}
&\mathbb{E}\big[S_{\SU }| s_{\SM}\big] \nonumber\\
&= 
(N-M)\frac{\mathcal{B}\big(1+(M+1)\varepsilon_{\mathrm E},(M+1)(1-\varepsilon_{\mathrm E})\big)}
{\mathcal{B}\big((M+1)\varepsilon_{\mathrm E},(M+1)(1-\varepsilon_{\mathrm E})\big)} \nonumber\\
&=(N-M)\frac{(M+1)\varepsilon_{\mathrm E}}{(M+1)\varepsilon_{\mathrm E}+(M+1)(1-\varepsilon_{\mathrm E})} \nonumber\\
&=(N-M)\varepsilon_{\mathrm E},\label{eqn:expKMbar}
\end{align}
where the second equality is obtained by exploiting the property of the beta function, i.e.,
\begin{align}
\mathcal{B}(\alpha+1, \beta) = \mathcal{B}(\alpha, \beta) \frac{\alpha}{\alpha+\beta}.\label{eqn:betaprop}
\end{align}

In \eqref{eqn:Mraw_2}, by setting $j=2$ and applying \eqref{eqn:betaprop}, one can obtain the second moment as follows.
\begin{align}
\mathbb{M}^{(2)}\big[S_{\SU }\big| s_{\SM}\big]&=(N-M)\varepsilon_{\mathrm E}\bigg(1+\nonumber\\
&\hspace{4.5mm} (N-M-1)\frac{(M+1)\varepsilon_{\mathrm E}+1}{M+2}\bigg).\label{eqn:M2KMbar}
\end{align}

According to \eqref{eqn:Mraw_2} and the relation between the raw and central moments \cite[eqn.(6)]{Oga:J16}, one can obtain the central moments of the posterior distribution of $S_{\SU }$.
\begin{align}
&\mathbb{M}_{\mathrm{c}}^{(t)}\big[S_{\SU }\big| s_{\SM}\big]\nonumber\\
&= \sum_{j=0}^{t}(-1)^j\bigg(\begin{array}{c}t\\j\end{array}\bigg)(N-M)^j\big(\varepsilon_{\mathrm E}\big)^j
\mathbb{M}^{(t-j)}\big[S_{\SU }\big| s_{\SM}\big].\label{eqn:MCK}
\end{align}
In \eqref{eqn:MCK}, by setting $t=2$, and substituting \eqref{eqn:expKMbar} and \eqref{eqn:M2KMbar} into the equation,  the  variance of the posterior distribution of $S_{\SU }$ is obtained as follows.
\begin{align}
\mathbb{V}\big[S_{\SU }| s_{\SM}\big]
&= 
\mathbb{M}^{(2)}\big[S_{\SU }\big| s_{\SM}\big] -2(N-M)\varepsilon_{\mathrm E}\nonumber\\
&\hspace{4.5mm} \cdot \mathbb{E}\big[S_{\SU }| s_{\SM}\big] + (N-M)^2\big(\varepsilon_{\mathrm E}\big)^2
\nonumber\\
&=\frac{(N-M)(N+1)}{M+2}\varepsilon_{\mathrm E}(1-\varepsilon_{\mathrm E}).\label{eqn:varKMbar}
\end{align}

Given the one-to-one correspondence between $\varepsilon_{\SM}$ and $s_{\SM}$,
\begin{align}
\mathbb{M}^{(t)}\big[{\Set E}_{\SU }\big| \varepsilon_{\SM}\big] = \frac{1}{(N-M)^t}\mathbb{M}^{(t)}\big[S_{\SU }\big| s_{\SM}\big].\label{eqn:Mscale}
\end{align}
By substituting \eqref{eqn:Mscale} into \eqref{eqn:expKMbar}, \eqref{eqn:MCK}, and \eqref{eqn:varKMbar}, equations
 \eqref{eqn:kmexp_ap}, \eqref{eqn:MCkbar_ap}, and \eqref{eqn:kmvar_ap} is obtained. 
 This completes the proof of \thref{lem:moments}.
\end{proof}

\subsection{Tail probability bound of $\bar{F}$}
\label{sec:tailprob}
In this subsection, we bound the tail probability of the average fidelity $\bar{F}$ to obtain an interval estimate that is reliable in scenarios with general noise.
\begin{LemA}[Tail probability bound of conditional expectation]\thlabel{lem:tailprob}
$\Set E$ and $F$ are bounded random variables defined on $\mathbb{R}$.
The expectation of $\Set E$ conditioned on $F=f$ satisfies 
\begin{align}
\mathbb{E}[a\mathcal{E}+b|f]=f,\label{eqn:meancondition_ap}
\end{align}
In this case, the tail probability of $F$ is bounded by
\begin{align}
\Pr\big[\big|F-a\mathbb{E}[\mathcal{E}]-b\big| \geq \delta  \big] &\leq 
\frac{a^{2t}}{\delta^{2t}}\mathbb{M}_{\mathrm{c}}^{(2t)}[\mathcal{E}], \quad \forall t\in \mathbb{Z}^+. \label{eqn:ftail_ap}
\end{align}
\end{LemA}

\begin{proof}
According to \eqref{eqn:meancondition_ap}, the expectations of $F$ and $\Set E$ are related by
\begin{align}
\mathbb{E}[F]&=\int_{-\infty}^{+\infty}f{\Set P}(f)df = a\int_{-\infty}^{+\infty}\mathbb{E}[\mathcal{E}|f]{\Set P}(f)df + b\nonumber\\ 
&= a\mathbb{E}[\mathcal{E}] +b.
\label{eqn:equalexp}
\end{align}
According to \eqref{eqn:meancondition_ap} and \eqref{eqn:equalexp}, the even-order central moments of $F$ are upper bounded as follows.
\begin{subequations}
\begin{align}
&\mathbb{M}^{(2t)}_{\mathrm{c}}[F] \nonumber\\
&= \mathbb{E}\big[(F-\mathbb{E}[F])^{2t}\big] \nonumber \\
&= \int_{-\infty}^{+\infty} \Big( \int_{-\infty}^{+\infty}(a \varepsilon+b)\Set{P}(\varepsilon|f) d\varepsilon -a\mathbb{E}[\mathcal{E}] -b \Big)^{2t} {\Set P}(f)df \nonumber\\
&= a^{2t}\int_{-\infty}^{+\infty} \Big( \int_{-\infty}^{+\infty}(\varepsilon - \mathbb{E}[\mathcal{E}])\Set{P}(\varepsilon|f) d\varepsilon\Big)^{2t} {\Set P}(f)df \nonumber\\
&\leq a^{2t}\int_{-\infty}^{+\infty} \Big(\int_{-\infty}^{+\infty}(\varepsilon - \mathbb{E}[\mathcal{E}])^{2t}\Set{P}(\varepsilon|f)  d\varepsilon\Big) {\Set P}(f) df\label{eqn:Jensen}\\
&= a^{2t}\int_{-\infty}^{+\infty} (\varepsilon - \mathbb{E}[\mathcal{E}])^{2t} \Big(\int_{-\infty}^{+\infty}\Set{P}(\varepsilon|f){\Set P}(f)  df\Big)  d\varepsilon \label{eqn:Fubini}\\
&= a^{2t}\int_{-\infty}^{+\infty} (\varepsilon - \mathbb{E}[\mathcal{E}])^{2t} \Set{P}(\varepsilon)  d\varepsilon \nonumber\\
&= a^{2t}\mathbb{M}_{\mathrm{c}}^{(2t)}[\Set E],\label{eqn:varboundend}
\end{align}\label{eqn:varbound}
\end{subequations}
where \eqref{eqn:Jensen} is obtained by applying  Jensen's inequality considering that functions $x^{2t}$, $\forall t\in \mathbb{Z}^+$, are convex, and \eqref{eqn:Fubini} is obtained by applying Fubini's theorem based on the fact that $\varepsilon$ is bounded.

According to \eqref{eqn:equalexp} and \eqref{eqn:varboundend}, by applying Chebyshev's inequality (extended to higher moments), \eqref{eqn:ftail_ap} is obtained.
This completes the proof of \thref{lem:tailprob}.
\end{proof}

Based on \thref{lem:moments}~and~\thref{lem:tailprob}, the following theorem provides an interval estimate of $\bar{F}$ in the presence of general noise.

\begin{ThmA}[Interval estimation]\thlabel{thm:cred_interval}
In cases with general noise, the average fidelity $\bar{F}$ lies in the interval ${\Set C}^{(T)}$ with probability at least $\alpha$ conditioned on the \ac{QBER} $\varepsilon_{\SM}$, where 
\begin{align}
{\Set C}^{(T)} =&\big[\tilde{f}-\delta^{(T)},  \tilde{f}+\delta^{(T)} \big],\label{eqn:cre_interval_cor_ap}
\end{align}
$T$ can be any even positive integer, and the radius of the credible interval, $\delta^{(T)}$, is defined by
\begin{align}
\delta^{(T)} = \min_{\scriptsize t\in \{1,2,\ldots,\frac{T}{2}\} }\bigg\{\frac{3}{2}\bigg(\frac{\mathbb{M}^{(2t)}_{\mathrm{c}}\big[\mathcal{E}_{\SU }\big|\varepsilon_{\SM} \big]}{1-\alpha}\bigg)^{\frac{1}{2t}} \bigg\},\label{eqn:deltaRstar_ap}
\end{align}
with $\mathbb{M}^{(2t)}_{\mathrm{c}}\big[\mathcal{E}_{\SU }\big|\varepsilon_{\SM}\big]$ equals to $\mathbb{M}^{(2t)}_{\mathrm{c}}$ defined in Protocol~\ref{alg:fidelityest_int}. 
\end{ThmA} 

\begin{proof}
 \cite[Lemma B.4]{Ruan:J23} shows that the \ac{QBER} $\mathcal{E}_{\SM}$ and the average fidelity of the measured qubit pairs, $\bar{f}_{\SM}$, have the following relation.
\begin{align}
\mathbb{E}\bigg[1-\frac{3\mathcal{E}_{\SM }}{2}\Big|\bar{f}_{\SM}\bigg]=\bar{f}_{\SM}.\label{eqn:ktoM-mea}
\end{align}
By applying similar analysis to the measurement outcomes of the thought experiment, i.e., $\{r_n, n\in \SU = \Set N \backslash \SM\}$, it can be seen that
\begin{align}
\mathbb{E}\bigg[1-\frac{3\mathcal{E}_{\SU }}{2}\Big|\bar{f},\varepsilon_{\SM} \bigg]=\bar{f},\label{eqn:ktoM}
\end{align}
where the expectation is taken over the state and the measurement outcomes of the unsampled qubit pairs.
According to \eqref{eqn:kmexp_ap} and \eqref{eqn:ktoM} , the expectation of the average fidelity $\bar{F}$ conditioned on the \ac{QBER} $\varepsilon_{\SM}$ can be expressed as
\begin{subequations}
\begin{align}
\bar{f}_{\mathrm{E}}&=\mathbb{E}[\bar{F}|\varepsilon_{\SM}]\nonumber\\
&=1-\frac{3}{2}\mathbb{E}\big[\mathbb{E}[\mathcal{E}_{\SU }|\bar{f},\varepsilon_{\SM}]\big|\varepsilon_{\SM}]\label{eqn:checkf_arb-a}\\
&=1-\frac{3}{2}\mathbb{E}[\mathcal{E}_{\SU }|\varepsilon_{\SM}]\label{eqn:checkf_arb-b}\\
&= 1- \frac{3}{2}\Big(\varepsilon_{\SM} + \frac{\frac{1}{2}-\varepsilon_{\SM}}{M+1}\Big)\\
&= \tilde{f},
\end{align}
\label{eqn:checkf_arb}
\end{subequations}
where \eqref{eqn:checkf_arb-b} holds because of the tower property of conditional expectation \cite[P.4]{Pit:15}.

Using \eqref{eqn:MCkbar_ap}, \eqref{eqn:ktoM}, and \eqref{eqn:checkf_arb}, and applying Lemma~\ref{lem:tailprob}, one can obtain that
\begin{align}
&\Pr\big[|\bar{F}-\bar{f}_{\mathrm{E}}|  \geq \delta \big| \varepsilon_{\SM} \big]
\nonumber \\
&\qquad\leq \bigg(\frac{3}{2\delta}\bigg)^{2t}\mathbb{M}^{(2t)}_{\mathrm{c}}\big[\mathcal{E}_{\SU }\big|\varepsilon_{\SM} \big], \quad \forall t\in \mathbb{Z}^+
\label{eqn:barFtail_general}
\end{align}
Let $1-\alpha=\bigg(\dfrac{3}{2\delta_t}\bigg)^{2t}\mathbb{M}^{(2t)}_{\mathrm{c}}\big[\mathcal{E}_{\SU }\big|\varepsilon_{\SM} \big]$, then 
\begin{align}
\delta_t = \frac{3}{2}\bigg(\frac{\mathbb{M}^{(2t)}_{\mathrm{c}}\big[\mathcal{E}_{\SU }\big|\varepsilon_{\SM} \big]}{1-\alpha}\bigg)^{\frac{1}{2t}}.
\end{align}
Let $\delta^{(T)} = \argmin_{\scriptsize t\in \mathbb{Z}^+,\, \scriptsize t\leq T }\big\{\delta_t
\big\}$, and substitute it into \eqref{eqn:barFtail_general}, we obtain
\begin{align}
\Pr\big[|\bar{F}-\bar{f}_{\mathrm{E}}|  \leq \delta^{(T)} \big| \varepsilon_{\SM} \big] &\geq \alpha.
\label{eqn:barFtail_general2}
\end{align}
According to \eqref{eqn:barFtail_general2},  \eqref{eqn:cre_interval_cor_ap} is proved. 

The even central moments $\mathbb{M}^{(2t)}_{\mathrm{c}}\big[\mathcal{E}_{\SU }\big|\varepsilon_{\SM} \big]$ can be calculated according to their definition. According to \eqref{eqn:kmexp_ap}, $\varepsilon_{\mathrm E}$ is the expectation of the posterior distribution of $\mathcal{E}_{\SU }$. 
Thus, the central moments of this posterior distribution equals to $\mathbb{M}^{(2t)}_{\mathrm{c}}$ defined in Protocol~\ref{alg:fidelityest_int}.
\end{proof}

Since the expressions for the higher-order moments of the posterior distribution of $\mathcal{E}_{\SU }$ are complicated, 
it is not easy to draw insights from the credible interval given in \thref{thm:cred_interval}.
We have the following result to illustrate the impact of general noise on the estimation accuracy.

\begin{CorA}[Expression of ${\Set C}^{(2)}$]\thlabel{cor:cred_interval_2}
Given $\varepsilon_{\SM}$, the average fidelity $\bar{F}$ lies in interval ${\Set C}^{(2)}$ with probability at least $\alpha$, where
\begin{align}
{\Set C}^{(2)} =\big[\tilde{f}-\delta^{(2)}, \tilde{f}+\delta^{(2)} \big],\label{eqn:cre_interval_cor_sim}
\end{align}
in which the radius of the credible interval
\begin{align}
\delta^{(2)}=  \sqrt{\frac{1}{1-\alpha}}\sqrt{\frac{N+1}{N-M}}\sqrt{\frac{(2\tilde{f}+1)(1-\tilde{f})}{2(M+2)}}.\label{eqn:delta-s}
\end{align}
\end{CorA}
\begin{proof} In \eqref{eqn:cre_interval_cor_ap}, we set $T=2$. In this case,
\begin{align}
{\Set C}^{(2)} =&\big[\tilde{f}-\delta^{(2)}, \tilde{f}+\delta^{(2)} \big],\label{eqn:C2}
\end{align}
where according to \eqref{eqn:kmvar_ap},
\begin{align}
\delta^{(2)}&= \frac{3}{2}\sqrt{\frac{\mathbb{V}\big[\mathcal{E}_{\SU }\big|\varepsilon_{\SM} \big]}{1-\alpha}}\nonumber
\\&= \sqrt{\frac{9(N+1)}{4(1-\alpha)(N-M)(M+2)}\varepsilon_{\mathrm E}(1-\varepsilon_{\mathrm E})}\nonumber
\\&= \sqrt{\frac{1}{1-\alpha}}\sqrt{\frac{N+1}{N-M}}\sqrt{\frac{(2\tilde{f}+1)(1-\tilde{f})}{2(M+2)}}.\label{eqn:delta}
\end{align}
This completes the proof of \thref{cor:cred_interval_2}.
\end{proof}

\section{Determine  the maximum computed moment}
\label{sec:Tstar}
The maximum computed moment $T$ balances the accuracy of the credible interval and the computational complexity.
The most accurate credible interval is defined as
\begin{align}
{\Set C}^{*} =&\big[\tilde{f}-\delta^*, \tilde{f}+\delta^* \big],\label{eqn:cre_interval_cor_opt}
\end{align}
where 
\begin{align}
\delta^* = \min_{\scriptsize T\mbox{ even} }\bigg\{\frac{3}{2}\bigg(\frac{\mathbb{M}^{(T)}_{\mathrm{c}}\big[\mathcal{E}_{\SU }\big|\varepsilon_{\SM} \big]}{1-\alpha}\bigg)^{\frac{1}{T}} \bigg\}.\label{eqn:deltastar}
\end{align}
In this case, the ideal choice of $T$ is the minimum number that achieves the most accurate credible interval, i.e.,
\begin{align}
T^* = \argmin_{\scriptsize T\mbox{ even}} \big\{{\Set C}^{(T)} ={\Set C}^{*}\big\}.\label{eqn:Tstar-2}
\end{align}

\eqref{eqn:Tstar-2} is difficult to solve because of the complicated expression of the central moments of the posterior distribution of $\Set{E}_{\SU}$, which is a beta-binomial distribution compressed by a factor $N-M$.
To overcome this difficulty, we consider the following properties of beta-binomial distributions.

The beta-binomial distribution of $\Set E_{\SU}$ is a composite distribution composed of a beta distribution with parameters $e_{\SM}+\frac{1}{2}$ and $M- e_{\SM}+\frac{1}{2}$ and a set of binomial distributions with $N-M$ Bernoulli trials and a probability of success drawn randomly from the beta distribution.
When 
\begin{align}
e_{\SM}\gg 1, \quad M- e_{\SM}\gg1, \quad \mbox{and}\quad N-M\gg1,\label{eqn:eMNgg1}
\end{align}
the beta distribution and the binomial distributions all approximate the normal distribution \cite{Das:B10}.
Since the composite distribution of normal distributions is still normal, when \eqref{eqn:eMNgg1} holds, it makes sense to simplify the posterior distribution of $\Set E_{\SU}$ to a normal distribution.
In the following analysis, we perform this simplification. 

Given that the posterior distribution of $\Set{E}_{\SU}$ is a normal distribution with standard deviation~$\sigma$,
the even central moments of this distribution are given by
\begin{align}
\mathbb{M}^{(T)}_{\mathrm{c}}\big[\Set{E}_{\SU}\big|\varepsilon_{\SM}\big] = \frac{(\sqrt{2}\sigma)^T}{\sqrt{\pi}}\Gamma\Big(\frac{T+1}{2}\Big).\label{eqn:MTnormal}
\end{align}
Substituting \eqref{eqn:MTnormal} into \eqref{eqn:deltastar} yields
\begin{align}
\delta^* = \min_{\scriptsize T\in \mathbb{Z}^+, 2|T }\bigg\{\frac{3\sigma}{\sqrt{2}}\bigg(\frac{\Gamma\big(\frac{T+1}{2}\big)}{\sqrt{\pi}(1-\alpha)}\bigg)^{\frac{1}{T}} \bigg\}.\label{eqn:minT_1}
\end{align}
By omitting the constant terms and taking the logarithm, one can see that the minimization problem in \eqref{eqn:minT_1} has the same optimal solution as
\begin{align}
\min_{\scriptsize T\in \mathbb{Z}^+, 2|T }\bigg\{\frac{1}{T}\bigg(
\ln\Big(\Gamma\Big(\frac{T+1}{2}\Big)\Big)- \ln\big({\sqrt{\pi}(1-\alpha)}\big) \bigg)\bigg\}.\label{eqn:minT_2}
\end{align}

Since $\frac{T+1}{2}\geq \frac{3}{2} >1$, one can apply Stirling's formula \cite{MukSon:J16} to accurately approximate the gamma function $\Gamma\big(\frac{T+1}{2}\big)$. 
In this case, \eqref{eqn:minT_2} can be rewritten as
\begin{align}
\min_{\scriptsize T\in \mathbb{Z}^+, 2|T }\bigg\{\frac{1}{T}\bigg(
\ln\Big(\sqrt{2\pi e}\Big(\frac{T-1}{2e}\Big)^{\frac{T}{2}}\Big)- \ln\big({\sqrt{\pi}(1-\alpha)}\big) \bigg)\bigg\},
\end{align}
which is equivalent to
\begin{align}
\min_{\scriptsize T\in \mathbb{Z}^+, 2|T }\bigg\{
\frac{1}{2}\ln(T-1)
+
\frac{1}{T}\Big(\frac{1}{2}\big(1+\ln(2)\big)- \ln(1-\alpha) \Big)\bigg\}.\label{eqn:minT_2}
\end{align}

Define the function
\begin{align}
g(T) = \frac{1}{2}\ln(T-1)+\frac{1}{T}l(\alpha),
\end{align}
where 
\begin{align}
l(\alpha) = \frac{1}{2}(1+\ln(2))- \ln(1-\alpha),
\end{align}
and extend the domain of $g$ to $T\in \mathbb{R}^+$.
In this case,
\begin{align}
\frac{\partial g}{\partial T} = \frac{1}{2(T-1)} - \frac{1}{T^2} l(\alpha).
\end{align}
To minimize $g(T)$, set $\frac{\partial f}{\partial T}=0$, which yields
\begin{align}
T^2 - 2l(\alpha)T + 2l(\alpha) = 0. \label{eqn:Tquad}
\end{align}
By solving \eqref{eqn:Tquad} and omitting the solution with $T<1$, the following expression is obtained:
\begin{align}
\begin{split}
T^* &= l(\alpha)\bigg(1+\sqrt{1-\frac{2}{l(\alpha)}}\bigg)\\
&\approx l(\alpha)\bigg(1+1-\frac{1}{l(\alpha)}\bigg)\\
&=-2\ln(1-\alpha) + \ln(2),
\end{split}\label{eqn:Tsolution}
\end{align}
where the approximation is obtained considering the Taylor expansion of $\sqrt{1-x}$ and maintaining its first two terms.

\eqref{eqn:Tsolution} is obtained considering $T\in \mathbb{R}^+$. 
Since $T$ is an even positive integer, the rounding function must be used.
Therefore, the solution of \eqref{eqn:minT_2} has the form
\begin{align}
T^* = 2\bigg\lfloor\frac{-2\ln(1-\alpha) + c}{2}\bigg\rceil,\label{eqn:Tstar-3}
\end{align}
where $c$ is set to compensate for errors caused by the rounding function.

\begin{figure*}[t] \centering
\includegraphics[scale=1]{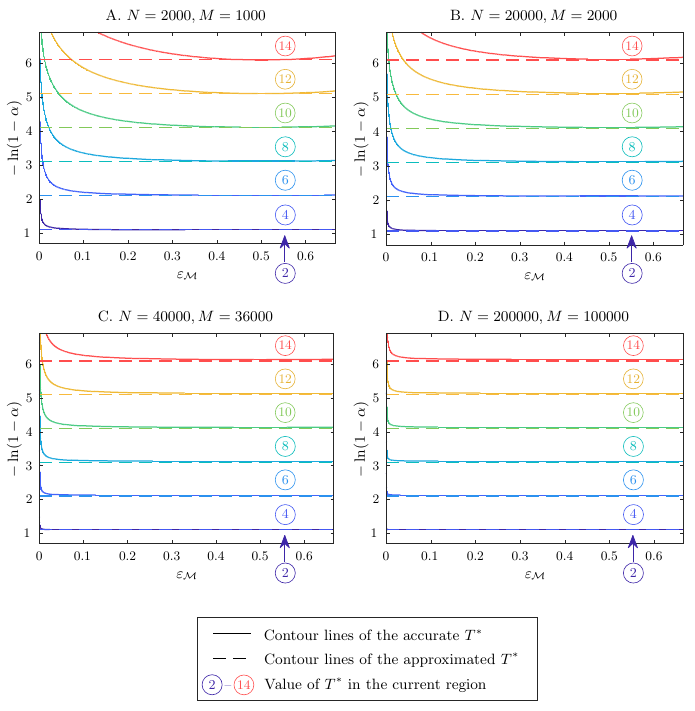}
\caption {Contour lines of the exact and approximated optimal computed moment $T^*$ in the plane spanned by the \ac{QBER} $\varepsilon_{\SM}$ and the credible probability $\alpha$ with four combinations of $N$ and $M$. 
This figure shows that $T^*$ given by \eqref{eqn:Tstar-3} is a tight upper bound on that given by \eqref{eqn:Tstar-2}.
}
\label{fig_Cre_Int_moments}
\end{figure*}

The value of $c$ is determined by comparing the exact value, i.e., $T^*$ given by \eqref{eqn:Tstar-2}, with the approximated value, i.e., $T^*$ given by \eqref{eqn:Tstar-3}, in a range of parameters.
More precisely, parameters satisfying
\begin{align*}
N&\in\{2000,4000,\ldots,200000\},\\
M&\in\{1000,2000,\ldots,100000\},\\
\alpha &\in \{0.5, 0.501, \ldots,0.999\}, \mbox{ and }\\
\varepsilon_{\SM}&\in\{0,0.001,\ldots,0.666\},
\end{align*}
are tested, and it turns out that \eqref{eqn:Tstar-3} is a tight upper bound of \eqref{eqn:Tstar-2} when $c$ is set to $0.8$.
Therefore, Protocol~\ref{alg:fidelityest_int} uses this setting.
Fig.~\ref{fig_Cre_Int_moments} shows the results of four representative tests.
For quick reference, Table~\ref{tab:optmoment} lists the value of $T^*$ given by \eqref{eqn:Tstar-3} for several typical credible probability values, $\alpha$.

\begin{table}[h]
\caption{The optimal computed moment $T^*$ for different $\alpha$ \label{tab:optmoment}}\vspace{2mm}
\centerline{
\begin{tabular}{C{1cm}"C{1cm}|C{1cm}|C{1cm}|C{1cm}|C{1cm}}
\hlinew{1pt}
$\alpha$& 0.80 & 0.90 & 0.95 & 0.98 & 0.99 \\ 
\hline
$T^*$&\Clb 4 &\Clb  6 &\Clb  6 &\Clb  8 &\Clb  10\\ 
\hlinew{1pt}
\end{tabular}
}
\end{table}

\section{Interval estimation with i.i.d.~noise}
\label{sec:iid}
In the case of \ac{i.i.d.}~noise, a critical value is the expected \ac{QBER} conditional on the average fidelity, i.e.,
\begin{align}
\varepsilon = \mathbb{E}\big[\mathcal{E}_{\SM}\big | \bar{f}\, \big].\label{eqn:k}
\end{align}
We will first derive the posterior distribution of $\Set E$ in \thref{lem:postk} and then obtain credible interval of $\bar {F}$ in \thref{thm:Bayid}. 
Here, $\Set E$ and $\bar{F}$ are the random variable form of $\varepsilon$ and $\bar{f}$, respectively.

\begin{LemA}[Posterior distribution of $\mathcal{E}$] 
\thlabel{lem:postk}
Given the number of errors measured $e_{\SM}$, the posterior distribution of $\mathcal{E}$ defined in \eqref{eqn:k}  is a beta distribution with parameters $e_{\SM}+\frac{1}{2}$ and $M- e_{\SM}+\frac{1}{2}$, i.e.,
\begin{align}
\Set{P}(\varepsilon | e_{\SM})&=\frac{\varepsilon^{e_{\SM}-\frac{1}{2}}(1-\varepsilon)^{M- e_{\SM}-\frac{1}{2}}}{\mathcal{B}(e_{\SM}+\frac{1}{2},M- e_{\SM}+\frac{1}{2})}.
\end{align}
\end{LemA}
\begin{proof}
With \ac{i.i.d.} noise, the measurement outcomes $r_n$, $n\in \Set M$, are \ac{i.i.d.} binary random variables.
Thus, given that the expected \ac{QBER} ${\Set E}=\varepsilon$, the number of measured errors, $E_{\Set M}$, follows a  binomial distribution. 
Therefore, the Jeffreys prior of $\Set E$
is given by \cite[P. 52]{GelCarSteDunVehRub:B14}
\begin{align}
\Set{P}(\varepsilon) &= \frac{\sqrt{\Set{I}(\varepsilon)}}{\int_0^1\sqrt{\Set{I}(\varepsilon)}dk}\nonumber\\
&=\frac{\varepsilon^{-\frac{1}{2}}(1-\varepsilon)^{-\frac{1}{2}}}{\int_0^1\varepsilon^{-\frac{1}{2}}(1-\varepsilon)^{-\frac{1}{2}}dk}=\frac{\varepsilon^{-\frac{1}{2}}(1-\varepsilon)^{-\frac{1}{2}}}{\mathcal{B}(\frac{1}{2},\frac{1}{2})}.\label{eqn:pk_prior}
\end{align}
where $\Set{I}(\cdot)$ denotes the Fisher information.

According to \eqref{eqn:pk_prior}, the prior distribution of $E_{\Set M}$ is given by
\begin{align}
&\Set{P}(e_{\Set M}) \nonumber\\
&= \int_0^1\Set{P}(e_{\Set M}|\varepsilon)\Set{P}(\varepsilon)d\varepsilon\nonumber\\
&=\int_0^1 \bigg(\begin{array}{c}M\\ e_{\Set M}\end{array}\bigg)
\varepsilon^{e_{\Set M}}(1-\varepsilon)^{M- e_{\Set M}}
\frac{\varepsilon^{-\frac{1}{2}}(1-\varepsilon)^{-\frac{1}{2}}}{\mathcal{B}(\frac{1}{2},\frac{1}{2})} d\varepsilon\nonumber\\
&=\bigg(\begin{array}{c}M\\ e_{\Set M}\end{array}\bigg)\frac{1}{\mathcal{B}(\frac{1}{2},\frac{1}{2})}
 \int_0^1 \varepsilon^{e_{\Set M}-\frac{1}{2}}(1-\varepsilon)^{M- e_{\Set M}-\frac{1}{2}} d\varepsilon\nonumber\\
&=\bigg(\begin{array}{c}M\\ e_{\Set M}\end{array}\bigg)\frac{\mathcal{B}(e_{\Set M}+\frac{1}{2},M- e_{\Set M}+\frac{1}{2})}{\mathcal{B}(\frac{1}{2},\frac{1}{2})}.\label{eqn:pKM_prior}
\end{align}
According to \eqref{eqn:pk_prior} and \eqref{eqn:pKM_prior}, by applying Bayes' theorem, the posterior distribution of $\mathcal{E}$ is determined as follows.
\begin{align}
&\Set{P}(\varepsilon| e_{\Set M})\nonumber\\
&=\frac{\Set{P}(e_{\Set M}|\varepsilon)\Set{P}(\varepsilon)}{\Set{P}(e_{\Set M})}\nonumber\\
&=\frac{\bigg(\begin{array}{c}M\\ e_{\Set M}\end{array}\bigg)
\varepsilon^{e_{\Set M}}(1-\varepsilon)^{M- e_{\Set M}}
\dfrac{\varepsilon^{-\frac{1}{2}}(1-\varepsilon)^{-\frac{1}{2}}}{\mathcal{B}(\frac{1}{2},\frac{1}{2})}}{\bigg(\begin{array}{c}M\\ e_{\Set M}\end{array}\bigg)\dfrac{\mathcal{B}(e_{\Set M}+\frac{1}{2},M- e_{\Set M}+\frac{1}{2})}{\mathcal{B}(\frac{1}{2},\frac{1}{2})}}\nonumber\\
&=\frac{\varepsilon^{e_{\Set M}-\frac{1}{2}}(1-\varepsilon)^{M- e_{\Set M}-\frac{1}{2}}}{\mathcal{B}(e_{\Set M}+\frac{1}{2},M- e_{\Set M}+\frac{1}{2})}.
\label{eqn:k_post}
\end{align}
This completes the proof of \thref{lem:postk}
\end{proof}

\begin{ThmA}[Interval estimation assuming \ac{i.i.d.} noise]\thlabel{thm:Bayid}
In cases with \ac{i.i.d.}~noise, $\bar{F}$ lies in interval ${\Set C}_{\mathrm{iid}}$ with probability $\alpha$, where
\begin{align}
{\Set C}_{\mathrm{iid}} = \bigg[&1-\frac{3}{2}I^{-1}_{1-\frac{\alpha}{2}}\Big(e_{\SM}+\frac{1}{2} , M- e_{\SM}+\frac{1}{2}\Big),
\nonumber\\
&1-\frac{3}{2}I^{-1}_{\frac{\alpha}{2}}\Big(e_{\SM}+\frac{1}{2} , M- e_{\SM}+\frac{1}{2}\Big)\bigg],
\label{eqn:Cid}
\end{align}
and $I^{-1}_{z}(a,b)$ denotes the inverse regularized incomplete beta function.
\end{ThmA}
\begin{proof}
In the case of \ac{i.i.d.} noise, one can see that
the average fidelity  $\bar{f}$ of the unsampled qubit pairs is a linear function of the expected \ac{QBER} $\varepsilon$, i.e.,
\begin{align}
\bar{f}=1- \frac{3\varepsilon}{2}.\label{eqn:barFkiid}
\end{align}

According to \thref{lem:postk} and the properties of beta distributions,
\begin{align}
\mathbb{E}\big[\mathcal{E} \big|\varepsilon_{\Set M} \big] 
&= \frac{e_{\Set M}+\frac{1}{2} }{e_{\Set M}+\frac{1}{2}  + M- e_{\Set M}+\frac{1}{2}} \nonumber\\
&=\varepsilon_{\Set M} + \frac{\frac{1}{2}-\varepsilon_{\Set M}}{M+1},\label{expek}
\end{align}
and
\begin{align}
\Pr\big[\mathcal{E}<k_0 \big|\varepsilon_{\Set M} \big] &= I_{k_0}\Big(e_{\Set M}+\frac{1}{2} , M- e_{\Set M}+\frac{1}{2}\Big),\label{Prklow}\\
\Pr\big[\mathcal{E}>k_1 \big|\varepsilon_{\Set M} \big] &= 1-I_{k_1}\Big(e_{\Set M}+\frac{1}{2}, M- e_{\Set M}+\frac{1}{2}\Big),\label{Prkhigh}
\end{align}
where $I_{x}(a,b)$ denotes the regularized incomplete beta function with argument $x$ and parameters $a$ and $b$.

According to  \eqref{eqn:barFkiid} and \eqref{expek}, the expected value of $\bar{F}$ conditioned on the estimation outcome is given by
\begin{align}
\bar{f}_{\mathrm{E}}&= \mathbb{E}\bigg[1-\frac{3\mathcal{E}}{2} \bigg|\varepsilon_{\Set M}\bigg]
= 1- \frac{3}{2}\Big(\varepsilon_{\Set M} + \frac{\frac{1}{2}-\varepsilon_{\Set M}}{M+1}\Big).
\end{align}
According to \eqref{eqn:barFkiid}, \eqref{Prklow}, and \eqref{Prkhigh}, it can be obtained that $\bar{F}$ is in interval ${\Set C}_{\mathrm{iid}}$
defined in \eqref{eqn:Cid} with probability $\alpha$.
This completes the proof of \thref{thm:Bayid}.
\end{proof}

The credible interval obtained in \thref{thm:Bayid} is relatively complicated. 
 The following corollary simplifies the expression of the credible interval in the asymptotic region with large $M$ to gain insight.

\begin{CorA}[Asymptotic credible interval]\thlabel{cor:cre_inter} 
In the case of \ac{i.i.d.}~noise and $M\gg 1$,
the credible interval of $\bar{F}$ simplifies as follows.
\begin{align}
{\Set C}_{\mathrm{iid}} = \big[\tilde{f}-\delta_{\mathrm{iid}},\tilde{f}+\delta_{\mathrm{iid}}\big],\label{eqn:Cid2}
\end{align}
where the radius of the interval $\delta_{\mathrm{iid}}$ is defined as
\begin{align}
\delta_{\mathrm{iid}} = Q\Big(\frac{1+\alpha}{2}\Big)\sqrt{\frac{(2\tilde{f}+1)(1-\tilde{f})}{2(M+2)}},
\end{align}
in which $Q(\cdot)$ is the quantile function of the standard normal distribution.
\end{CorA}

\begin{proof}
According to \cite[Sec. 10.7]{Das:B10} and \thref{lem:postk}, when $M\gg 1$, the posterior distribution of the average \ac{QBER} $\mathcal{E}$ can be approximated by a normal distribution with mean
\begin{align}
\mathbb{E}\big[\mathcal{E} \big|\varepsilon_{\Set M} \big]= \frac{e_{\Set M}+\frac{1}{2} }{e_{\Set M}+\frac{1}{2} + M- e_{\Set M}+\frac{1}{2}} = \varepsilon_{\mathrm E}\label{eqn:normalmean}
\end{align}
and variance
\begin{align}
\mathbb{V}\big[\mathcal{E} \big|\varepsilon_{\Set M} \big]&=\frac{\varepsilon_{\mathrm E}(1-\varepsilon_{\mathrm E})}{e_{\Set M}+\frac{1}{2} + M- e_{\Set M}+\frac{1}{2}+1} \nonumber\\
&= \frac{\varepsilon_{\mathrm E}(1-\varepsilon_{\mathrm E})}{M+2},\label{eqn:normalvar}
\end{align}
where $ \varepsilon_{\mathrm E}=\varepsilon_{\Set M} + \frac{\frac{1}{2}-\varepsilon_{\Set M}}{M+1}$.
According to \eqref{eqn:normalmean} and \eqref{eqn:normalvar}, the posterior distribution of $\bar{F}$ is a normal distribution with mean
\begin{align}
\mathbb{E}\big[\bar{F} \big|\varepsilon_{\Set M} \big] = 1 -\frac{3}{2}\mathbb{E}\big[\mathcal{E} \big|\varepsilon_{\Set M} \big] = \tilde{f}
\end{align}
and variance 
\begin{align}
\mathbb{V}\big[\bar{F} \big|\varepsilon_{\Set M} \big] &=\Big(\frac{3}{2}\Big)^2\mathbb{V}\big[\mathcal{E} \big|\varepsilon_{\Set M} \big]\nonumber
\\ &=\frac{(2\tilde{f}+1)(1-\tilde{f})}{2(M+2)}.
\end{align}
This distribution leads to the credible interval defined in \eqref{eqn:Cid2}.
This completes the proof of \thref{cor:cre_inter}.
\end{proof}

\bibliographystyle{quantum}
\bibliography{Quantum_Ruan_V4}

\end{document}